\documentclass[twocolumn]{IEEEtran}

\input{mysymbol.sty}
\usepackage{graphicx,slashbox}
\usepackage{epsfig}
\usepackage{amsmath}
\usepackage{bbm,bm,dsfont}
\usepackage{amssymb}
\usepackage{subfigure}
\usepackage{wrapfig}
\usepackage{times,color}
\usepackage{cite,footnote,xspace,syntonly,algorithm,algorithmic}
\usepackage[caption=false,font=footnotesize]{subfig}
\usepackage[T1]{fontenc}
\usepackage[cal=rsfs,bb=esstix,scr=boondoxo,frak=euler]{mathalfa}

\def \cN {\mathcal{N}}
\def \cL {\mathcal{L}}

\def \cS {\mathcal{S}}

\def \cP {\mathcal{P}}

\def \cT {\mathcal{T}}

\def \cR {\mathcal{R}}

\def \bY {{\bf Y}}
\def \bR {{\bf R}}
\def \bX {{\bf X}}
\def \bA {{\bf A}}
\def \bV {{\bf V}}
\def \bW {{\bf W}}
\def \bF {{\bf F}}
\def \bG {{\bf G}}
\def \bI {{\bf I}}
\def \bB {{\bf B}}
\def \bM {{\bf M}}
\def \bU {{\bf U}}
\def \bV {{\bf V}}
\def \bE {{\bf E}}

\def \bZ {{\bf Z}}
\def \bW {{\bf W}}
\def \bH {{\bf H}}
\def \bL {{\bf L}}
\def \bQ {{\bf Q}}
\def \bC {{\bf C}}
\def \bO {{\bf O}}
\def \bP {{\bf P}}

\def \bGamma {{\boldsymbol \Gamma}}

\def \bh {{\bf h}}
\def \bw {{\bf w}}

\def \bv {{\bf v}}
\def \bx {{\bf x}}
\def \bz {{\bf z}}
\def \bb {{\bf b}}
\def \bc {{\bf c}}
\def \be {{\bf e}}

\def \bv {{\bf v}}
\def \bo {{\bf o}}
\def \ba {{\bf a}}
\def \bm {{\bf m}}
\def \by {{\bf y}}
\def \bl {{\bf l}}
\def \bq {{\bf q}}
\def \bxi {{\boldsymbol \xi}}
\def \balpha {{\boldsymbol \alpha}}
\def \bdelta {{\boldsymbol \Phi}}

\def \bPi {{\boldsymbol \Pi}}
\def \bSigma {{\boldsymbol \Sigma}}

\def \sgn {\text{sgn}}
\def \tr {\text{tr}}

\def \supp {\text{supp}}
\def \Phiper {\Phi^{\bot}}
\def \Phiperx {\Phi_{X_0}^{\bot}}
\def \Omegaper {\Omega^{\bot}}
\def \Omegapera {\Omega_{A_0}^{\bot}}

\long\def\symbolfootnote[#1]#2{\begingroup
\def\thefootnote{\fnsymbol{footnote}}
\footnote[#1]{#2}\endgroup} \psfull

\begin{document}
\title{\huge Estimating Traffic and Anomaly Maps via Network Tomography$^\dag$}

\author{{\it Morteza Mardani and Georgios~B.~Giannakis~(contact author)$^\ast$}}

\markboth{IEEE/ACM TRANSACTIONS ON NETWORKING (SUBMITTED)}
\maketitle \maketitle \symbolfootnote[0]{$\dag$ Work in this
paper was supported by the MURI Grant No. AFOSR FA9550-10-1-0567.
Parts of the paper were presented in the {\it Proc. of
the IEEE International Conference on Acoustics, Speech, and Signal Processing}, Vancouver, Canada, May 26-31, 2013, and
in {\it IEEE Global Signal and Information Processing Workshop}, Austin, Texas,
December 3-5, 2013.} \symbolfootnote[0]{$\ast$ The authors are with the Dept.
of ECE and the Digital Technology Center, University of
Minnesota, 200 Union Street SE, Minneapolis, MN 55455. Tel/fax:
(612)626-7781/625-4583; Emails:
\texttt{\{morteza,georgios\}@umn.edu}}


\vspace{-0.5cm}


\thispagestyle{empty}\addtocounter{page}{-1}
\begin{abstract}
Mapping origin-destination (OD) network traffic is pivotal for network management and proactive security tasks. However, lack of sufficient flow-level measurements as well as potential anomalies pose major challenges towards this goal. Leveraging the spatiotemporal correlation of nominal traffic, and the sparse nature of anomalies, this paper brings forth a novel framework to map out nominal and anomalous traffic, which treats jointly important network monitoring tasks including traffic estimation, anomaly detection, and traffic interpolation. To this end, a convex program is first formulated with nuclear and $\ell_1$-norm regularization to effect sparsity and low rank for the nominal and anomalous traffic with only the link counts and a {\it small} subset of OD-flow counts. Analysis and simulations confirm that the proposed estimator can {\em exactly} recover sufficiently low-dimensional nominal traffic and sporadic anomalies so long as the routing paths are sufficiently ``spread-out'' across the network, and an adequate amount of flow counts are randomly sampled. The results offer valuable insights about data acquisition strategies and network scenaria giving rise to accurate traffic estimation. For practical networks where the aforementioned conditions are possibly violated, the inherent spatiotemporal traffic patterns are taken into account by adopting a Bayesian approach along with a bilinear characterization of the nuclear and $\ell_1$ norms. The resultant nonconvex program involves quadratic regularizers with correlation matrices, learned systematically from (cyclo)stationary historical data. Alternating-minimization based algorithms with provable convergence are also developed to procure the estimates. Insightful tests with synthetic and real Internet data corroborate the effectiveness of the novel schemes.

\end{abstract}

\begin{keywords}
Sparsity, low rank, convex optimization, nominal and anomalous traffic, spatiotemporal correlation.
\end{keywords}
%

%


\section{Introduction}
\label{sec:intro}
Emergence of multimedia services and Internet-friendly portable devices is multiplying network traffic volume day by day~\cite{wu11}. Moreover, the advent of diverse networks of intelligent devices including those deployed to monitor the smart power grid, transportation networks, medical information networks, and cognitive radio networks, will transform the communication infrastructure to an even more complex and heterogeneous one. Thus, ensuring compliance to service-level agreements necessitates ground-breaking management and monitoring tools providing operators with informative depictions of the network state. One such atlas (set of maps) can offer a flow-time depiction of the network origin-destination (OD) flow traffic. Situational awareness provided by such maps will be the key enabler for effective routing and congestion control, network health management, risk analysis, security assurance, and proactive network failure prevention. Acquiring such diagnosis/prognosis maps for large networks however is an arduous task. This is mainly because the number of OD pairs grows promptly as the network size grows, while probing exhaustively all OD pairs becomes impractical even for moderate-size networks~\cite{shavit01}. In addition, OD flows potentially undergo anomalies arising due to e.g., cyberattacks and network failures~\cite{LPC04}, and the acquired measurements typically encounter misses, outliers, and errors.

Towards creating traffic maps, one typically has access to: (D1) link counts comprising the superposition of OD flows per link; these counts can be readily obtained using the single network management protocol (SNMP)~\cite{LPC04}; and (D2) {\it partial} OD-flow counts recorded using e.g., the NetFlow protocol~\cite{LPC04}. Extensive studies of backbone Internet Protocol (IP) networks reveals that the nominal OD-flow traffic is spatiotemporally correlated mainly due to common temporal patterns across OD flows, and exhibits periodic trends (e.g., daily or weekly) across time~\cite{LPC04}. This renders the nominal traffic having a small intrinsic dimensionality. Moreover, traffic volume anomalies rarely occur across flows and time~\cite{LPC04,barford01,Papagiannaki03}. Given the observations (D1) and/or (D2), ample research has been carried out over the years to tackle the ill-posed traffic inference task relying on various techniques that leverage the traffic features as prior knowledge; see e.g.,~\cite{Kolaczyk_book,robusttrafficestimation_Zhao06,cascetta,vardi,zuylen,rzwq11,tit_exactrecovery_2012,zggr05} and references therein.

To date, the main body of work on traffic inference relies on least-squares (LS) and Gaussian~\cite{cascetta,robusttrafficestimation_Zhao06} or Poisson models~\cite{vardi}, and entropy regularization~\cite{zuylen}. None of these methods however takes spatiotemporal dependencies of the traffic into account. To enhance estimation accuracy by exploiting the spatiotemporal dependencies of traffic, attempts have been made in~\cite{rzwq11}~and~\cite{tit_exactrecovery_2012}. Using the prior spatial and temporal structures of traffic, \cite{rzwq11} applies rank regularization along with matrix factorization to discover the global low-rank traffic matrix from the link and/or flow counts. The model in~\cite{rzwq11} is however devoid of anomalies, which can severely deteriorate traffic estimation quality. In the context of anomaly detection, our companion work~\cite{tit_exactrecovery_2012} capitalizes on the low-rank of traffic and sparsity of anomalies to unveil the traffic volume anomalies from the link loads (D1). Without OD-flow counts however, the nominal flow-level traffic cannot be identified using the approach of~\cite{tit_exactrecovery_2012}.

The present work addresses these limitations by introducing a novel framework that efficiently and scalably constructs network traffic maps. Leveraging recent advances in compressive sensing and rank minimization, first, a novel estimator is put forth, to effect sparsity and low rank attributes for the anomalous and nominal traffic components through $\ell_1$- and nuclear-norm, respectively. The recovery performance of the sought estimator is then analyzed in the noise-free setting following a deterministic approach along the lines of~\cite{CSPW11}. Sufficient incoherence conditions are derived based on the angle between certain subspaces to ensure the retrieved traffic and anomaly matrices coincide with the true ones. The recovery conditions yield valuable insights about the network structures and data acquisition strategies giving rise to accurate traffic estimation. Intuitively, one can expect accurate traffic estimation if:~(a) NetFlow measures sufficiently many randomly selected OD flows; (b) the OD paths are sufficiently ``spread-out'' so as the routes form a column-incoherent routing matrix; (c) the nominal traffic is sufficiently low dimensional; and, (d) anomalies are sporadic enough.

Albeit insightful, the accurate-recovery conditions in practical networks may not hold. For instance, it may happen that a specific flow undergoes a bursty anomaly lasting for a long time~\cite{barford01}, or certain OD flows may be inaccessible for the entire time horizon of interest with no NetFlow samples at hand. With the network practical challenges however come opportunities to exploit certain structures, and thus cope with the aforementioned challenges. This work bridges this ``theory-practice'' gap by incorporating the spatiotemporal patterns of the nominal and anomalous traffic, both of which can be learned from historical data. Adopting a Bayesian approach, a novel estimator is introduced for the traffic following a bilinear characterization of the nuclear- and $\ell_1$-norms. The resultant nonconvex problem entails quadratic regularizers loaded with inverse correlation matrices to effect structured sparsity and low rank for anomalous and nominal traffic matrices, respectively. A systematic approach for learning traffic correlations from historical data is also devised taking advantage of the (cyclo)stationary nature of traffic. Alternating majorization-minimization algorithms are also developed to obtain iterative estimates, which are provably convergent.

Simulated tests with synthetic network and real Internet-data corroborate the effectiveness of the novel schemes, especially in reducing the number of acquired NetFlow samples needed to attain a prescribed estimation accuracy. In addition, the proposed optimization-based approach opens the door for efficient in-network and online processing along the lines of our companion works in~\cite{tsp_rankminimization_2012} and~\cite{jstsp_anomalography_2012}. The novel ideas can also be applicable to various other inference tasks dealing with recovery of structured low-rank and sparse matrices.

The rest of this paper starts with preliminaries and problem statement in Section~\ref{sec:model}. The novel estimator to map out the nominal and anomalous traffic is discussed in Section~\ref{sec:estimator}, and pertinent reconstruction claims are established in Section~\ref{sec:rec_guarantees}. Sections~\ref{sec:inc_cor} and~\ref{sec:bayes_est} deal with incorporating the spatiotemporal patterns of traffic to improve estimation quality. Certain practical issues are addressed in Section~\ref{sec:pr_issues}. Simulated tests are reported in Section~\ref{sec:perf_eval}, and finally Section~\ref{sec:conclusions} draws the conclusions.



\noindent\textit{Notation:}~Bold uppercase (lowercase) letters will denote matrices (column vectors), and calligraphic letters will be used for sets. Operators $(\cdot)'$, $\rm{tr}(\cdot)$, $\sigma_{\max}(\cdot)$, $[\cdot]_+$, $\oplus$, $\odot$ and $\mathbb{E}[\cdot]$, ${\rm dim}(\cdot)$ will denote transposition, matrix trace, maximum singular-value, projection onto the nonnegative orthant, direct sum, Hadamard product, statistical expectation, and dimension of a subspace, respectively; $|\cdot|$ will stand for cardinality of a set, and the magnitude of a 
scalar.  The $\ell_p$-norm of $\bx \in \mathbb{R}^n$ is $\|\bx\|_p:=(\sum_{i=1}^n |x_i|^p)^{1/p}$ for $p \geq 1$. For two matrices $\bM,\bU \in \mathbb{R}^{n \times n}$,
$\langle \bM, \bU \rangle := \rm{tr(\bM' \bU)}$ denotes their trace inner product. The Frobenius norm of matrix $\bM = [m_{i,j}] \in \mathbb{R}^{n \times p}$ is $\|\bM\|_F:=\sqrt{\tr(\bM\bM')}$,
$\|\bM\|:=\max_{\|\bx\|_2=1} \|\bM\bx\|_2$ is the spectral norm,
and $\|\bA\|_{\infty}:=\max_{i,j} |a_{ij}|$
the $\ell_{\infty}$-norm. The $n \times n$ identity matrix will be represented by $\bI_n$ and its $i$-th column by $\be_i$, while $\mathbf{0}_{n}$ will stand for 
the $n \times 1$ vector of all zeros, $\mathbf{0}_{n \times p}:=\mathbf{0}_{n} \mathbf{0}'_{p}$. Operator ${\rm vec}$ stacks columns of a matrix, and conversely does ${\rm unvec}$; $\cap$ and $\cup$ stand for the set intersection and union, respectively; ${\rm supp}(\bA):=\{(i,j):a_{ij} \neq 0\}$ is the support set of $\bA$, and $[n]:=\{1,\ldots,n\}$.


\section{Preliminaries and Problem Statement}
\label{sec:model}
Consider a backbone IP network described by the directed graph 
$G(\cal{N},\cal{L})$, where $\mathcal{L}$ and $\cal{N}$ denote the set of 
links and nodes (routers) of cardinality $|\mathcal{L}|=L$ and 
$|\mathcal{N}|=N$, respectively. A set of end-to-end flows $\cal{F}$ with 
$|\mathcal{F}| = F$ traverse different OD pairs. In backbone networks, the 
number of OD flows far exceeds the number of physical links $(F \gg L)$. Per 
OD-flow, multipath routing is considered where each flow traverses multiple 
possibly overlapping paths to reach its intended destination. Letting 
$x_{f,t}$ denote the unknown traffic level of flow $f \in \mathcal{F}$ at time 
$t$, link $\ell \in \mathcal{L}$ carries the fraction $r_{\ell,f} \in [0,1]$ 
of this flow; clearly, $r_{\ell,f}=0$ if flow $f$ is not routed through link 
$\ell$. The total traffic carried by link $\ell$ is then the weighted 
superposition of flows routed through link $\ell$, that is, 
$\sum_{f\in\mathcal{F}}r_{\ell,f}x_{f,t}$. The weights $\{r_{\ell,f}\}$ form 
the routing matrix $\bR \in [0,1]^{L \times F}$, which is assumed fixed and 
given. These weights are not arbitrary but must respect the flow conservation 
law $\sum_{\ell \in \mathcal{L}_{\rm in}(n)} r_{\ell,f} = \sum_{\ell \in 
\mathcal{L}_{\rm out}(n)} r_{\ell,f},~ \forall f \in \mathcal{F}$, where 
$\mathcal{L}_{\rm in}(n)$ and $\mathcal{L}_{\rm out}(n)$ denote the sets of 
incoming and outgoing links to node $n\in\mathcal{N}$, respectively.

It is not uncommon for some of flow rates to experience sudden changes, which 
are termed \textit{traffic volume anomalies}
that are typically due to the network failures, or cyberattacks
\cite{LPC04}. With $a_{f,t}$ denoting the unknown traffic volume
anomaly of flow $f$ at time $t$, the traffic carried by
link $\ell$ at time $t$ is
\begin{equation}
y_{\ell,t}=\sum_{f\in\cal{F}}r_{\ell,f}(x_{f,t}+a_{f,t}) + v_{\ell,t},\quad t\in\mathcal{T} \label{eq:y_lt}
\end{equation}
where the time horizon $\mathcal{T}$ comprises $T$ slots, and $v_{\ell,t}$ accounts for the measurement errors. In IP networks, link loads can be readily measured via SNMP supported by most routers~\cite{LPC04}. Introducing the matrices $\bY:=[y_{\ell,t}],
\bV:=[v_{\ell,t}]\in\mathbbm{R}^{L \times T}$, $\bX:=[x_{f,t}]$, and
$\bA:=[a_{f,t}] \in\mathbbm{R}^{F \times T}$, link counts
in~\eqref{eq:y_lt} can be expressed in a compact matrix form as
\begin{equation}
\bY=\bR \left(\bX + \bA\right) + \bV \label{eq:Y}.
\end{equation}
Here, matrices $\bX$ and $\bA$ contain, respectively, the {\it nominal} and
{\it anomalous} traffic flows over the time horizon $\mathcal{T}$. Inferring $(\bX,\bA)$ from the compressed measurements $\bY$ is a severely underdetermined task (recall that $L \ll
F$), necessitating additional data to ensure identifiability and improve estimation accuracy. A useful such source is the direct flow-level
measurements
\begin{equation}
z_{f,t}=x_{f,t}+a_{f,t} + w_{f,t},~t\in\mathcal{T},~f \in \mathcal{F} \label{eq:z_lt}
\end{equation}
where $w_{f,t}$ accounts for measurement errors. The flow traffic in~\eqref{eq:z_lt} is sampled via NetFlow~\cite{LPC04} at each origin
node. This however incurs high cost which means that one can have measurements~\eqref{eq:z_lt}~only for few $(f,t)$ pairs~\cite{LPC04}. To account for
missing flow-level data, collect the available pairs $(f,t)$ in the
set $\Pi \in [F] \times [T]$; introduce also the
matrices $\bZ_{\Pi}:=[z_{f,t}],\bW_{\Pi}:=[w_{f,t}] \in
\mathbbm{R}^{F \times T}$, where $z_{f,t}=w_{f,t}=0$ for $(f,t)
\notin \Pi$, and associate the sampling operator $\cP_{\Pi}$
with the set $\Pi$, which assigns entries of its matrix argument
not in $\Pi$ equal to zero, and keeps the rest unchanged. As with $\bX$, it holds that $\cP_{\Pi}(\bX) \in \mathbbm{R}^{F \times T}$. The
flow counts in~\eqref{eq:z_lt} can then be compactly written as
\begin{equation}
\bZ_{\Pi}=\cP_{\Pi} \left(\bX + \bA \right) + \bW_{\Pi} \label{eq:Z}.
\end{equation}
Besides periodicity, temporal patterns common to traffic flows render rows (correspondingly columns) of
$\bX$ correlated, and thus $\bX$ exhibits a few dominant singular values which 
make it (approximately) low rank~\cite{LPC04}. Anomalies on the 
other hand are expected to occur occasionally, as only a small fraction of 
flows are supposed to be anomalous at any given time instant, which means $\bA$ is sparse. Anomalies may exhibit certain patterns e.g., failure 
at a part of the network may simultaneously render a subset of flows 
anomalous; or certain flows may be subject to bursty malicious attacks over 
time.

Given the link counts $\bY$ obeying~\eqref{eq:Y} along with the partial
flow-counts $\bZ_{\Pi}$ adhering to~\eqref{eq:Z}, and with
$\{\bR,\Pi\}$ known, this paper aims at accurately estimating the
unknown {\it low-rank} nominal and {\it sparse} anomalous traffic pair $(\bX,\bA)$.

\begin{figure}
\centering
\begin{tabular}{c}
     \epsfig{file=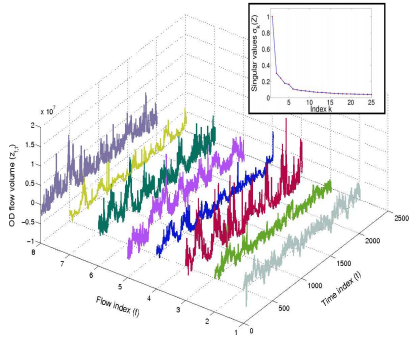,width=0.6\linewidth,height=1.75 in}  \\
  \end{tabular}
  \caption{Internet-2 traffic for a few representative OD flows across time and flows~\cite{Internet2}. }
  \label{fig:fig_lowrank_traffic}
\end{figure}

\section{Maps of Nominal and Anomalous Traffic}
\label{sec:estimator}
In order to estimate the unknowns of interest, a natural estimator accounting for the low rank of $\bX$
and the sparsity of $\bA$ will be sought to minimize the rank of
$\bX$, and the number of nonzero entries of $\bA$ measured by its
$\ell_0$-(pseudo) norm. Unfortunately, both rank and $\ell_0$-norm
minimization problems are in general
NP-hard~\cite{Natarajan_NP_duro,RFP07,CT05}. The nuclear-norm
$\|\bX\|_*:=\sum_{k}\sigma_{k}(\bX)$, where $\sigma_k (\bX)$ signifies
the $k$-th singular value of $\bX$, and the $\ell_1$-norm
$\|\bA\|_1:=\sum_{f,t}|a_{f,t}|$ are typically adopted
as~\textit{convex} surrogates ~\cite{RFP07,CT05}. Accordingly, one
solves
\begin{align}
\text{(P1)}~~(\hat{\bX},\hat{\bA} )&=\arg\min_{(\bX,\bA)} \frac{1}{2}\|\bY - \bR(\bX+\bA)\|_F^2 \nonumber \\ 
\hspace{4mm} &+ \frac{1}{2}\|\cP_{\Pi}(\bZ-\bX-\bA)\|_F^2 + \lambda_{\ast}\|\bX\|_{\ast} + \lambda_1\|\bA\|_1
\nonumber
\end{align}
where $\lambda_1,\lambda_{\ast}\geq 0$ are the sparsity- and rank-controlling parameters. From a network operation perspective, the estimate $\hat{\bA}$ maps out the network health-state across both time and flows. A large value $|\hat{a}_{f,t}|$ indicates that at time instant $t$ flow $f$ exhibits a sever anomaly, and therefore appropriate traffic engineering and security tasks need to be run to mitigate the consequences. The estimated map of nominal traffic $\hat{\bX}$ is also a viable input for network planning tasks.

From the recovery standpoint, (P1) subsumes several important special
cases, which deal with recovery of $\hat{\bX}$ and/or $\hat{\bA}$. In the absence of flow counts, i.e.,~$\Pi=\emptyset$, exact recovery of the \emph{sparse} anomaly matrix $\hat{\bA}$ from link loads is established
in~\cite{tit_exactrecovery_2012}. The key to this is the sparsity
present, which enables recovery from compressed linear-measurements. However, the (possibly huge) nullspace of $\bR$ challenges
identifiability of the nominal traffic matrix $\bX$, as will be delineated later. Moreover, with only flow counts partially available, (P1) boils down to the
so-termed robust principal component pursuit (PCP), for which exact reconstruction of the \emph{low-rank} nominal traffic component is established in~\cite{CSPW11}. Instrumental role in this case is played by the dependencies among entries of the low-rank component, reflected in the observations. Indeed, the matrix of anomalies is not recoverable since observed entries do not convey any information about the unobserved anomalies. Furthermore, without the sparse matrix, i.e.,~$\bA=\mathbf{0}$, and only with flow counts partially available, (P1) boils down to the celebrated matrix completion problem studied e.g., in \cite{candes_moisy_mc}, which can be applied to interpolate the traffic of unreachable OD flows from the observed ones at the edge routers.

The aforementioned considerations regarding recovery in these special cases make one hopeful to retrieve $\bX$ and $\bA$ via
(P1). Before delving into the analysis of (P1), it is worth noting that~\cite{compressive_pcp_ganesh} has recently studied recovery of compressed low-rank-plus-sparse matrices, also known as compressive PCP, where the
compression is performed by an orthogonal projection onto a low-dimensional
subspace, and the the support of the sparse matrix is presumed uniformly random. The results require certain subspace incoherence conditions to hold, which in the considered traffic estimation task impose strong restrictions on the routing matrix $\bR$ and the sampling operator $\cP_{\Pi}(\cdot)$. Furthermore, it is unclear how to relate the subspace incoherence conditions to the well-established incoherence measures adopted in the context of matrix completion and compressive sampling, which are satisfied by various classes of random matrices; see e.g.,~\cite{CT05,candes2009exact}.

Before closing this section, it is important to recognize that albeit few the NetFlow measurement $\bZ_{\Pi}$, they play an important role in estimating $\bX$. In principle, if one merely knows the link counts $\bY$, it is impossible to accurately identify $\bX$ when the only prior information about $\bX$ and $\bA$ is that they are sufficiently low-rank and sparse, respectively. This identifiability issue is formalized in the next lemma.

\begin{lemma}{\label{lemma:lem_1}} 
With $\cN_{R}$ denoting the nullspace of $\bR$, and $\bX_0=\bU_0\bSigma_0\bV_0^{'}$, if $\cN_R \neq \emptyset$, and one only knows $\{\bY,\bR\}$, then for any $\bW \in \cN_R$ the matrix pair $\{\bX_1:=\bX_0+\bW\bV_0',\bA_0\}$: (i) is feasible, and (ii) it satisfies ${\rm rank}(\bX_1) \leq {\rm rank}(\bX_0)=:r$.
\end{lemma}

\begin{IEEEproof}
Clearly (i) holds true since $\bR\bW=\mathbf{0}$, and subsequently $\bR(\bA_0+\bX_1)=\bR(\bA_0+\bX_0)+\bR\bW\bV_0'=\bY$. Also, (ii) readily follows from Sylvester's inequality~\cite{Horn} which implies that ${\rm rank}(\bU_0\bSigma_0\bV_0^{'}+\bW\bV_0^{'}) \leq \min\{{\rm rank}(\bX_0+\bW\bV_0^{'}),{\rm rank}(\bV_0)\} \leq {\rm rank}(\bV_0)=r$.
\end{IEEEproof}
%

%
%


\section{Reconstruction Guarantees}
\label{sec:rec_guarantees}
This section studies the exact reconstruction performance of~(P1) in the absence of noise, namely $\bV=\mathbf{0}$ and $\bW_{\Pi}=\mathbf{0}$. The corresponding formulation can be expressed as 
\begin{align}
\text{(P2)}~~~&(\hat{\bX},\hat{\bA} )=\arg\min_{(\bX,\bA)}~~ \|\bX\|_* + \lambda \|\bA\|_1 \nonumber \\ & {\rm s. to} \quad \bY=\bR \left(\bX + \bA\right), \quad \bZ_{\Pi}
=\cP_{\Pi} \left(\bX + \bA \right). \nonumber
\end{align}
In the sequel, identifiability of $(\bX,\bA)$ from the linear measurements $\{\bY,\bZ_{\Pi}\}$ is pursued first, followed by technical conditions based on certain incoherence measures, to guarantee $(\hat{\bX}=\bX_0,\hat{\bA}=\bA_0)$, where $\bX_0$ and $\bA_0$ are the {\it true} low-rank and sparse matrices of interest.

\subsection{Local Identifiability}
\label{subsec:identifiability}
Let $r:={\rm rank}(\bX_0)$ and $s:=\|\bA_0\|_0$ denote the rank and sparsity level of the true matrices of interest. The first issue to address is identifiability, asserting that there is a {\it unique} pair $(\bX_0,\bA_0)$ fulfilling the data constraints: (d1)
$\bY=\bR(\bX_0+\bA_0)$ and (d2)
$\bZ_{\Pi}=\cP_{\Pi}(\bX_0+\bA_0)$. Apparently, if multiple
solutions exist, one cannot hope finding $(\bX_0,\bA_0)$. Before
examining this issue, introduce the subspaces: (s1) $\cN_{R}:=\{\bH:
\bR\bH=\mathbf{0}_{L \times T}\}$ as the nullspace of the linear
operator $\bR$, and (s2) $\cN_{\Pi}:=\{\bH \in
\mathbbm{R}^{F\times T}: \supp(\bH) \subseteq \Pi^{\bot}\}$ as
the nullspace of the linear operator $\cP_{\Pi}(.)$ [$\Pi^{\bot}$ is the complement of $\Pi$]. If
there exists a perturbation pair $(\bH_1,\bH_2)$ with $\bH_1+\bH_2 \in
\cN_R \cap \cN_{\Pi}$ so that $\bX_0+\bH_1$ and $\bA_0+\bH_2$ are
still low-rank and sparse, one may pick the pair
$(\bX_0+\bH_1,\bA_0+\bH_2)$ as another legitimate solution. This
section aims at resolving such identifiability issues.

Let $\bU_0\bSigma_0\bV'_0$ denote the singular value decomposition (SVD) of
$\bX_0$, and consider the subspaces: (s3) $\Phi_{X_0} :=
\{\bZ\in\mathbbm{R}^{F\times T}:\bZ=\bU_0 \bW_1' + \bW_2 \bV'_0,\:\bW_1 \in
\mathbbm{R}^{T \times r},\:\bW_2 \in {\mathbbm{R}}^{F \times r}\}$ of
matrices in either the column or row space of $\bX_0$; (s4)
$\Omega_{A_0}:=\{\bH\in\mathbbm{R}^{F\times T}:\supp(\bH) \subseteq
\supp(\bA_0)\}$ of matrices whose support is contained in that of $\bA_0$.
Noteworthy properties of these subspaces are: (i) since $\Phi_{X_0}$ and
$\Omega_{A_0} \subset \mathbbm{R}^{F\times T}$, it is possible to
directly compare elements from them; (ii) $\bX_0 \in \Phi_{X_0}$ and $\bA_0
\in \Omega_{A_0}$; and (iii) if $\bZ\in\Phiperx$ is added to $\bX_0$, then
$\textrm{rank}(\bZ+\bX_0)>r$, and likewise $\bZ\in\cN_{\Omega}$,~for any $\bZ \in \Omega_{A_0}^{\bot}$.

Suppose temporarily that the subspaces $\Phi_{X_0}$ and
$\Omega_{A_0}$ are also known. This extra piece of information helps
identifiability based on data (d1) and (d2) since the potentially
troublesome solutions
\begin{align}
\Upsilon_1 : = \{(\bX_0 + \bH_1,\bA_0+\bH_2): \bH_1+\bH_2 \in \cN_{R} \cap
\cN_{\Pi}\} \label{eq:subspace_f1}
\end{align}
are restricted to a smaller set. If $(\bX_0+\bH_1,\bA_0+\bH_2) \notin
\Upsilon_2$, where
\begin{align}
\Upsilon_2 : = \{(\bX_0 + \bH_1,\bA_0+\bH_2): \bH_1 \in \Phi_{X_0},~  \bH_2 \in \Omega_{A_0}\} \label{eq:subspace_f2}
\end{align}
that candidate solution is not admissible since it is known a priori
that $\bX_0 \in \Phi_{X_0}$ and $\bA_0 \in \Omega_{A_0}$. This
notion of exploiting additional knowledge to assure uniqueness is
known as {\it local identifiability}~\cite{CSPW11}. Global
identifiability from (d1) and (d2) is not guaranteed. However, local
identifiability will become essential later on to establish the main
result. With these preliminaries, the following lemma puts forth the
necessary and sufficient conditions for local identifiability.

\begin{lemma}{\label{lemma:lem_2}}
Matrices $(\bX_0,\bA_0)$ satisfy (d1) and (d2) uniquely if and only if~(c1)~$\Phi_{X_0} \cap \Omega_{A_0} = \{\mathbf{0}\}$; and,~(c2)~$\Upsilon_1 \cap
\Upsilon_2= \{\mathbf{0}\}$.
\end{lemma}

Condition (c1) implies that for the solutions in $\Upsilon_2$ to be
admissible, $\bH_1+\bH_2$ must belong to the subspace $\Phi_{X_0}
\oplus \Omega_{A_0}$. Accordingly, (c2) holds true if
\begin{align}
\cN_{R} \cap \cN_{\Pi} \cap (\Phi_{X_0} \oplus \Omega_{A_0}) = \{\mathbf{0}
\}. \label{eq:localiden_iffcond}
\end{align}
Notice that (c1) appears also in the context of low-rank-plus-sparse
recovery results in~\cite{CLMW09,CSPW11}. However, (c2) is unique to
the setting here. It captures the impact of the overlap between the nullspace of $\bR$ and the operator $\cP_{\Pi}(\cdot)$.
Finding simpler sufficient conditions to assure (c1) and (c2) is dealt with next.

\subsection{Incoherence Measures}
\label{subsec:incoherence_measures}
The overlap between any pair of subspaces
$\{\Phi_{X_0},\Omega_{A_0},\cN_R,\cN_{\Pi}\}$ plays a crucial
role in identifiability and exact recovery as seen e.g., from
Lemma~\ref{lemma:lem_1}. To quantify the overlap of the subspaces
e.g.,~$\Phi_{X_0}$ and $\Omega_{A_0}$, consider the {\it
incoherence} parameter
\begin{align}
\mu(\Phi_{X_0},\Omega_{A_0}):= \max_{ \substack{\bX \in \Omega_{A_0} \\
\|\bX\|_F=1} }
\|\cP_{\Phi_{X_0}}(\bX)\|_F, \label{eq:incoherence_add}
\end{align}
which clearly satisfies $\mu(\Phi_{X_0},\Omega_{A_0}) \in [0,1]$. The lower
bound is achieved when $\Phi_{X_0}$ and $\Omega_{A_0}$ are
orthogonal, whereas the upperbound is attained when
$\Phi_{X_0}\cap\Omega_{A_0}$ contains a nonzero element. To gain
further geometric intuition, $\mu(\Phi_{X_0},\Omega_{A_0})$
represents the cosine of the angle between subspaces
when they have trivial intersection, namely~$\Phi_{X_0} \cap \Omega_{A_0} = \{\mathbf{0}\}$~\cite{Deutsch}.
Small values of $\mu(\Phi_{X_0},\Omega_{A_0})$ indicate sufficient
separation between $\Phi_{X_0}$ and $\Omega_{A_0}$, and thus less
chance of ambiguity when discerning $\bX_0$ from $\bA_0$.

It will be seen later that (c1) requires
$\mu(\Phi_{X_0},\Omega_{A_0}) <1$. In addition, to ensure (c2) one needs the
incoherence parameter $\mu(\cN_R \cap \cN_{\Pi}, \Phi_{X_0} \oplus
\Omega_{A_0}) < 1$. In fact, $\mu(\cN_R \cap \cN_{\Pi}, \Phi_{X_0} \oplus
\Omega_{A_0})$ captures the ambiguity inherent to the nullspace of the
compression and sampling operators. It depends on all subspaces (s1)--(s4), and
it is desirable to express it in terms of the incoherence of different
subspace pairs, namely~$\mu(\cN_R,\Omega_{A_0})$, $\mu(\cN_R,\Phi_{X_0})$,
$\mu(\cN_{\Pi},\Omega_{A_0})$, and $\mu(\cN_{\Pi},\Phi_{X_0})$. This is formalized in the next claim.

\begin{proposition}\label{prop:prop_1}
Assume that $\mu(\Omega_{A_0},\Phi_{X_0}) < 1$. If either ${\rm dim}(\cN_{R}
\cap \cN_{\Pi}) =0$; or, ${\rm dim}(\cN_{R} \cap \cN_{\Pi}) \geq 1$ and
\begin{align}
\hspace{-1mm}\chi:=\Big[\frac{\mu(\cN_{\Pi},\Phi_{X_0})
+ \mu(\cN_R,\Omega_{A_0})\mu(\cN_{\Pi},\Omega_{A_0})}
{1-\mu(\Omega_{A_0},\Phi_{X_0})}\Big]^{1/2} \hspace{-2mm}< 1  \nonumber
\end{align}
hold, then $\Phi_{X_0} \cap \Omega_{A_0} = \{\mathbf{0}\}$ and $\cN_{R} \cap
\cN_{\Pi} \cap (\Phi_{X_0} \oplus \Omega_{A_0}) = \{\mathbf{0} \}$.
\end{proposition}
\begin{IEEEproof}
Since $\mu(\Omega_{A_0},\Phi_{X_0}) < 1$ and $\rm {dim}(\Phi_{X_0} \oplus \Omega_{A_0} \oplus (\cN_{R} \cap
\cN_{\Pi})) = \rm {dim}(\Phi_{X_0}) + \rm {dim}(\Omega_{A_0}) + \rm{dim}(\cN_{R} \cap
\cN_{\Omega})$,~\cite[Lemma~11]{compressive_pcp_ganesh} implies that 
\begin{align}
\mu^2(\Phi_{X_0} \oplus \Omega_{A_0}, \cN_{R} \cap \cN_{\Pi}) \leq \big[1- \mu(\Phi_{X_0},\Omega_{A_0}) \big]^{-1} \nonumber \\ \times \big[ \mu^2(\Phi_{X_0},\cN_{R} \cap \cN_{\Pi}) + \mu^2(\Omega_{A_0},\cN_{R} \cap \cN_{\Pi}) \big]. 
\end{align}
The result then follows by bounding~$\mu^2(\Phi_{X_0},\cN_{R} \cap \cN_{\Pi}) \leq \mu(\Phi_{X_0},\cN_{R}) \mu(\Phi_{X_0},\cN_{\Pi})$~using the fact that $\cN_{R} \cap \cN_{\Pi} \in  \cN_{R}, \cN_{\Pi}$ [likewise for $\mu(\Omega_{A_0},\cN_{R} \cap \cN_{\Pi})$], and $\cN_R \cap \cN_{\Phi_{X_0}} \neq \{\mathbf{0}\}$.
\end{IEEEproof}


Apparently, small values of $\mu(\cN_R,\Omega_{A_0})$ and
$\mu(\cN_{\Pi},\Phi_{X_0})$ gives rise to a small $\chi$. In
fact, $\mu(\cN_R,\Omega_{A_0})$ measures whether
$\cN_R$ contains sparse elements, and it is tightly related to the
incoherence among the sparse column-subsets of $\bR$. For row-orthonormal compression matrices in particular, where $\bR\bR'=\bI$, the incoherence reduces to the restricted isometry
constant of $\bR$, see e.g.,~\cite{CT05}. Moreover,
$\mu(\cN_{\Pi},\Phi_{X_0})$ measures whether the low-rank
matrices fall into the nullspace of the subsampling operator $\cP_{\Pi}(\cdot)$, that is tightly linked to the incoherence metrics introduced in the context of
matrix completion; see e.g.,~\cite{CR08}. It is worth
mentioning that a wide class of matrices resulting in small
incoherence $\mu(\cN_R,\Omega_{A_0})$,
$\mu(\cN_{\Pi},\Phi_{X_0})$ and $\mu(\Omega_{A_0},\Phi_{X_0})$
are provided in~\cite{CT05},~\cite{CR08},~\cite{CLMW09}, which give
rise to a sufficiently small value of $\chi$.


\subsection{Exact Recovery via Convex Optimization}
\label{subsec:exactrecovery} Besides $\mu (\Omega_{A_0},\Phi_{X_0})$
and $\chi$, there are other incoherence measures which play an
important role in the conditions for exact recovery. These measures
are introduced to avoid ambiguity when the (feasible) perturbations
$\bH_1$ and $\bH_2$ do not necessarily belong to the subspaces
$\Phi_{X_0}$ and $\Omega_{A_0}$, respectively. Before moving on, it
is worth noting that these measures resemble the ones for matrix
completion and decomposition problems; see e.g.,~\cite{CLMW09,CR08}.
For instance, consider a feasible solution
$\{\bX_{0}+a_{i,j}\be_i\be_j',\bA_0 + a_{i,j}\be_i\be_j'\}$, where
$(i,j) \notin \supp(\bA_0)$, and thus $a_{i,j}\be_i\be_j' \notin
\Omega_{A_0}$. It may happen that $a_{i,j}\be_i\be_j' \in \Phi_{X_0}$
and ${\rm rank} (\bX_{0}+a_{i,j}\be_i\be_j') = {\rm
rank}(\bX_{0})-1$, while $\|\bA_0-a_{i,j}\be_i\be_j'\|_0 =
\|\bA_0\|_0 + 1$, thus challenging identifiability when $\Phi_{X_0}$ and
$\Omega_{A_0}$ are unknown. Similar complications arise if $\bX_0$
has a sparse row space that can be confused with the row space of
$\bA_0$. These issues motivate defining
\begin{align}
\gamma(\bU_0):=\max_{i}\|\bP_{U}\be_i\|,\quad \quad
\gamma(\bV_0):=\max_{i} \|\bP_{V}\be_i\| \label{eq:gamma_u}
\end{align}
where $\bP_{U}:=\bU_0\bU_0'$~(resp.~$\bP_{V}:=\bV_0\bV_0'$) are the
projectors onto the column (row) space of $\bX_0$. Notice that
$\gamma(\bU_0),\gamma(\bV_0) \in [0,1]$. The maximum of
$\gamma(\bU_0)$~(resp.~$\gamma(\bV_0)$) is attained when $\be_i$ is in the
column (row) space of $\bX_{0}$ for some $i$. Small values of
$\gamma(\bU_0)$~(resp.~$\gamma(\bV_0)$) imply that the column~(row) spaces
of $\bX_{0}$ do not contain sparse vectors, respectively.

Another identifiability instance arises when $\bX_0$ is sparse, in which case each column of $\bX_{0}$ is spanned by a few canonical basis vectors. Consider the parameter
\begin{align}
\gamma(\bU_0,\bV_0) := \|\bU_0\bV'_0\|_{\infty} =
\max_{i,j}|{\be_i}'\bU_0\bV_0\be_j|. \label{eq:gamma_uv}
\end{align}
A small value of $\gamma(\bU_0,\bV_0)$ indicates that each column of
$\bX_{0}$ is spanned by sufficiently many canonical basis vectors. It is worth
noting that $\gamma(\bU_0,\bV_0)$ can be bounded in terms of $\gamma(\bU_0)$
and $\gamma(\bV_0)$, but it is kept here for the sake of generality.

From (c2) in Lemma~\ref{lemma:lem_1} it is evident that the dimension of the
nullspace $\cN_R \cap \cN_{\Pi}$ is critical for identifiability. In
essence, the lower dim($\cN_R \cap \cN_{\Pi}$) is, the higher is the chance for exact reconstruction. In order to quantify the size of the nullspace,
define
\begin{align}
\tau(\cN_R,\cN_{\Pi}):=\max_{ \substack{\bX \in \cN_R \cap \cN_{\Pi}
 \\ \|\bX\|=1} } \|\bX\|_{\infty} \label{eq:tau_infty}
\end{align}
which will appear later in the exact recovery conditions. All elements are
now in place to state the main result.

\subsection{Main Result}
\label{subsec:mainresults}

\begin{theorem}\label{th:theorem_1}
Let $(\bX_0,\bA_0)$ denote the true low-rank and sparse matrix pair of interest, and define $\bX_0:=\bU_0\bSigma_0\bV'_0$, $r:=\text{rank}(\bX_{0})$, and $s:=\|\bA_0\|_0$. Assume that $\bA_0$ has at most $k$ nonzero elements per column, and define the incoherence parameters $\alpha:=\mu(\Omega_{A_0},\Phi_{X_0})$, $\beta:=\mu(\Omega_{A_0},\cN_R)$, $\xi:=\mu(\cN_{\Pi},\Phi_{X_0})$, $\nu:=\mu(\cN_R,\Omega_{A_0} \cap \cN_{\Pi})$, $\eta:=\gamma(\bU_0)+\gamma(\bV_0)$, $\tau:=\tau(\cN_R,\cN_{\Pi})$,
$\gamma:=\gamma(\bU_0,\bV_0)$. Given $\bY$ and $\bZ_{\Pi}$ adhering to~(d1) and~(d2), respectively, with known $\bR$ and $\Pi$, if $\chi < 1$, and
\begin{align}
&{\rm (I)}\quad \lambda_{\max}:= (\frac{1}{k})\frac{1-\alpha - \alpha^3(1-\alpha^2)-ge/f}{1+\alpha^2(1-\alpha^2)+he/f} \nonumber \\ 
&\hspace{2cm} > \lambda_{\min}:= \frac{\gamma + qg/f}{1-\eta \alpha k - k qh/f} \geq 0 \nonumber\\
&{\rm (II)} \quad f:= 1-\nu \beta - (\xi + \alpha \nu) (1-\alpha^2)(\xi + \alpha \beta) > 0 \nonumber
\end{align}
hold, where
\begin{align}
g: = \xi + \alpha (\xi + \alpha \nu)(1-\alpha^2) \alpha , \nonumber \quad
h: = \nu + \alpha (1-\alpha^2)(\xi + \alpha \nu) \nonumber \\
q:= \tau + \eta \alpha + \eta \xi, \nonumber \quad
e:=  \alpha (1-\alpha^2)(\xi+\alpha \beta) + 1 + \nu \nonumber
\end{align}
then for any $\lambda_{\min} \leq \lambda \leq \lambda_{\max}$ the convex program (P1) yields $(\hat{\bX}=\bX_0,\hat{\bA}=\bA_0)$.
\end{theorem}

Satisfaction of the conditions in Theorem~\ref{th:theorem_1} hinges upon the
incoherence parameters $\{\alpha,\gamma,\eta,\xi,\tau\}$ whose sufficiently
small values fulfil (I) and (II). In fact, these parameters are increasing
functions of the rank $r$ and the sparsity level $s$. In particular,
$\{\alpha,\gamma,\eta\}$ that capture the ambiguity of the additive components
$\bX_0$ and $\bA_0$, are known to be small enough for small values of
$\{r,s,k\}$; see e.g., \cite{CR08,CSPW11}. Regarding $\chi$, recall that it is an
increasing function of $\beta$ and
$\xi$, where the parameter
$\xi$ takes a small value when NetFlow samples an
adequately large subset of OD flows uniformly at random. Moreover, in
large-scale networks with distant OD node pairs, and routing paths that are
sufficiently ``spread-out'', the sparse column-subsets of $\bR$ tend to be
incoherent, and thus $\beta$ takes a small value.
Likewise, for sufficiently many NetFlow samples and column-incoherent routing
matrices, $\tau$ takes a small value.

\begin{remark}[\textbf{Satisfiability}]
\normalfont Notice that (I) and (II) in Theorem~\ref{th:theorem_1} are expressible in terms of the angle between subspaces (s1)--(s4). In general, they are NP-hard to verify. Introducing a class of (possibly random) traffic matrices~$(\bX_0,\bA_0)$ and realistic network settings giving rise to a desirable routing matrix $\bR$ is the subject of our ongoing research. The major roadblock in this direction is deriving tight bounds for the parameter $\tau$, which involves the intersection of a pair of subspaces. 
\end{remark}

\subsection{ADMM Algorithm}
\label{subsec:admm_alg}
This section introduces an iterative solver for the convex program (P2) using the alternating direction method of multipliers (ADMM) method. ADMM is an iterative augmented Lagrangian method especially
well-suited for parallel processing~\cite{Bertsekas_Book_Distr}, and has been
proven successful to tackle the optimization tasks
encountered e.g., in statistical learning; see e.g.,~\cite{boyd_monograph_admom}.
While ADMM could be directly applied to (P2), $\bR$ couples
the entries of $\bA$ and $\bX$ leading to computationally demanding nuclear- and
$\ell_1$-norm minimization subtasks per iteration. To overcome this hurdle, a common
trick is to introduce auxiliary
(decoupling) variables $\{\bB,\bO\}$, and formulate the following optimization problem
\begin{align}
\text{(P3)} \quad \min_{\{\bA,\bX,\bO,\bB\}}& \|\bX\|_{*}+\lambda\|\bA\|_1 \nonumber \\
\text{s. to}~&\bY = \bR(\bO + \bB),\quad \bZ_{\Pi}=\cP_{\Pi} \left(\bO + \bB \right)   \nonumber \\
&\bB=\bA, \quad \bO=\bX, \nonumber
\end{align}
which is equivalent to (P2). To tackle (P3), associate the Lagrange
multipliers $\{\bM_y, \bM_z, \bM_a, \bM_x\}$ with the
constraints, and then introduce the
quadratically {\it augmented} Lagrangian function
\begin{align}
&\cL(\bX,\bA,\bB,\bO;\bM_y,\bM_z,\bM_a, \bM_x) \nonumber\\
&:=\|\bX\|_{*} + \lambda
\|\bA\|_1 + \langle \bM_y,\bY - \bR(\bO + \bB) \rangle + \langle 
\bM_a,
\bB-\bA \rangle \nonumber \\ 
&+ \langle
\bM_z,\bZ_{\Pi} - \cP_{\Pi} \left(\bO + \bB \right)\rangle + \langle \bM_x, \bO-\bX \rangle \nonumber \\ &
  + \frac{c}{2}\|\bY - \bR(\bO +
 \bB)\|_{F}^{2} + \frac{c}{2} \|\bZ_{\Pi} - \cP_{\Pi} \left(\bO + \bB
 \right)\|_{F}^{2} \nonumber \\ 
 &  + \frac{c}{2} \|\bB - \bA \|_{F}^{2} + \frac{c}{2} \|\bO -
 \bX \|_{F}^{2} \label{eq:lag_admom}
\end{align}
where $c>0$ is a penalty coefficient. Splitting the primal variables
into two groups $\{\bX,\bB\}$ and $\{\bA,\bO\}$, the ADMM solver entails
an iterative procedure comprising three steps per iteration $k = 1, 2,\ldots$

\begin{description}
 \item [{\bf [S1]}]  \textbf{Update dual variables:}
\begin{align}
    &\hspace{-0.4cm}\bM_y[k]=\bM_y[k-1]+c(\bY - \bR(\bO[k] + \bB[k])) \label{eq:M_y}\\
    &\hspace{-0.4cm}\bM_z[k]=\bM_z[k-1]+c(\bZ_{\Pi} - \cP_{\Pi} \left(\bO + \bB \right))
    \label{eq:M_z}\\
    &\hspace{-0.4cm}\bM_a[k]=\bM_a[k-1]+c(\bB[k] - \bA[k]) \label{eq:M_a}\\
    &\hspace{-0.4cm}\bM_x[k]=\bM_x[k-1]+c(\bO[k] - \bX[k]) \label{eq:M_x}
\end{align}

 \item [{\bf [S2]}]  \textbf{Update first group of primal variables:}
    \begin{align}
    &\hspace{-0.8cm}\bA[k+1]{}\nonumber \\
    &\hspace{-0.8cm}={} \arg \hspace{-2mm} \min_{\bA  \in \mathbbm{R}^{F \times T}} \left\{ \frac{c}{2}
    \|\bA-\bB[k]\|_{F}^{2}  - \langle
    \bM_a[k],\bA \rangle +
     \lambda \|\bA\|_1 \right\}.\nonumber \\
        &\hspace{-0.8cm}\bO[k+1]{} \nonumber \\ 
        &\hspace{-0.8cm}={}
        \arg\hspace{-2mm}\min_{\bO  \in \mathbbm{R}^{F \times T}} \left\{ \frac{c}{2}\|\bO - \bX[k]\|_{F}^{2} +
        \frac{c}{2} \|\bY - \bR(\bO + \bB[k])\|_{F}^{2} \right. \nonumber \\ & \hspace{-0.75cm}  \hspace{1.5cm} \left. + \frac{c}{2}
        \|\bZ_{\Pi} - \cP_{\Pi} \left(\bO + \bB[k] \right)\|_{F}^{2} \right. \nonumber \\ 
        &\hspace{-0.8cm}\hspace{1.5cm} \left. + \langle
        \bM_x[k]-\bR'\bM_y[k]-\cP_{\Pi}(\bM_z[k]),\bO \rangle \right\}.\nonumber
    \end{align}

\item [{\bf [S3]}]  \textbf{Update second group of primal variables:}
        \begin{align}
    &\hspace{-0.8cm}\bX[k+1]{} \nonumber \\ 
    &\hspace{-0.8cm}={} \mbox{arg}\hspace{-2mm}\min_{\bX \in \mathbbm{R}^{F \times T}} \left\{ \frac{c}{2}\|\bX - 
    \bO[k]\|_{F}^{2}   - \langle \bM_x[k], \bX \rangle
    + \|\bX\|_{\ast} \right\} \nonumber \\
     &\hspace{-0.8cm}\bB[k+1]{} \nonumber \\ 
     &\hspace{-0.8cm}={} \mbox{arg}\hspace{-2mm}\min_{\bB  \in \mathbbm{R}^{F \times T}} \left\{ \frac{c}{2}
     \|\bA[k]-\bB\|_{F}^{2} + \frac{c}{2} \|\bY - \bR(\bO[k] + \bB)\|_{F}^{2} \right.
     \nonumber \\ 
     & \hspace{-0.8cm} \hspace{1.5cm} \left.+ \frac{c}{2} \|\bZ_{\Pi} - \cP_{\Pi} \left(\bO[k] + \bB 
     \right)\|_{F}^{2} \right.
      \nonumber\\ &\hspace{-0.8cm} \hspace{1.5cm}  \left. + \langle \bM_a[k] -
     \bR'\bM_y[k] - \cP_{\Pi}(\bM_z[k]), \bB \rangle
     \right\} \nonumber
        \end{align}

\end{description}

The resulting iterative solver is tabulated under Algorithm~\ref{tab:admom_alg}. Here, $[\cS_{\tau}(\bX)]_{i,j}:={\sgn(x_{i,j})}\max\{|x_{i,j}|-\tau,0\}$ refers to the soft-thresholding operator;~the vectors $\{\by_t,\bo_t,\ba_t,\bb_t,\bz_t,\bx_t,\bm_t^z,\bm_t^a,\bm_t^x,\bm_t^y\}$ denote the $t$-th column of their corresponding matrix arguments, and the diagonal matrix $\bPi_t \in \{0,1\}^{P \times P}$ is unity at $(i,i)$-th entry if $(i,t) \in \Pi$, and zero otherwise. Algorithm~\ref{tab:admom_alg} reveals that the update for the anomaly matrix entails a soft-thresholding operator to promote sparsity, while the nominal traffic is updated via singular value thresholding to effect low rank. The updates for $\bB$ and $\bO$ are also parallelized across the rows. Due to convexity of (P3), Algorithm~\ref{tab:admom_alg} with two Gauss-Seidel block updates is convergent to the global optimum of (P2) as stated next.

\begin{proposition}\cite{Bertsekas_Book_Distr}
For any value of the penalty coefficient $c>0$, the
iterates $\{\bX[k],\bA[k]\}$ converge to the optimal solution of (P2) as $k
\rightarrow \infty$.
\end{proposition}

\begin{algorithm}[t]
\caption{: ADMM solver for (P2)} \small{
\begin{algorithmic}
	\STATE \textbf{input} $\bY, \bZ_{\Pi}, \Pi, \bR, \lambda, c, \{\bH_t:=(\bI_F+\bPi_t+\bR'\bR)^{-1}\}_{t=1}^T$
    \STATE \textbf{initialize} $\bM_y[-1]=\mathbf{0}_{L\times T}$,
        $\bbX[0]=\bbO[0]=\bbA[0]=\bbB[0]=\bM_z[-1]=\bM_a[-1]=\bM_x[-1]=\mathbf{
        0}_{F\times
        T}$, and set
        $k=0$.
    \WHILE {not converged}

        \STATE \textbf{[S1] Update dual variables:}

        \STATE $\bM_y[k]=\bM_y[k-1]+c(\bY - \bR(\bO[k] + \bB[k]))$
        \STATE $\bM_z[k]=\bM_z[k-1]+c(\bZ_{\Pi} - \cP_{\Pi} \left(\bO[k] 
        + \bB[k]
        \right))$
        \STATE $\bM_a[k]=\bM_a[k-1]+c(\bB[k] - \bA[k])$
        \STATE $\bM_x[k]=\bM_x[k-1]+c(\bO[k] - \bX[k])$

        \STATE \textbf{[S2] Update first group of primal variables:}

                \STATE 
                $\bA[k+1]=\cS_{\frac{\lambda}{c}}(c^{-1}\bM_a[k]+\bB[k])$.

        \STATE Update in parallel~($t=1,\ldots,T$)


        \STATE $\bo_t[k+1] = \bH_t \big( c \bx_t[k] + c \bPi_t 
        \bz_t + c \bR' \by_t - c [\bPi_t + \bR'\bR]\bb_t[k] + \bR' 
        \bm_t^y[k] + \bPi_t \bm_t^z[k]  - \bm_t^x[k]\big)$


        \STATE \textbf{[S3] Update second group of primal variables:}

        \STATE $\bU\bSigma\bV'= {\rm svd}(\bO[k+1] + c^{-1}\bM_x[k]), \quad 
        \bX[k+1]=\bU \mathcal{S}_{1/c} (\bSigma)\bV'$

        \STATE Update in parallel~($t=1,\ldots,T$)


        \STATE $\bb_t[k+1] = \bH_t \big( c \ba_t[k+1] + c 
        \bPi_t \bz_t + c\bR'\by_t  - c[\bPi_t + \bR'\bR]\bo_t[k+1] + 
        \bR'\bm_t^y[k] + \bPi_t \bm_t^z[k] - \bm_t^a[k]\big)$


		\STATE $k\leftarrow k+1$
    \ENDWHILE
    \RETURN $(\bA[k],\bX[k])$
\end{algorithmic}}
\label{tab:admom_alg}
\end{algorithm}

\section{Incorporating Spatiotemporal Correlation Information}
\label{sec:inc_cor}

Being convex (P1) is appealing, and as Theorem~\ref{th:theorem_1} asserts for the noiseless case it reconstructs reliably the underlying traffic when:~(c1) the anomalous traffic is sufficiently ``sporadic'' across time and flows; (c2) the nominal traffic 
matrix is sufficiently low-rank with non-spiky singular vectors; (c3) NetFlow 
{\it uniformly} samples OD flows; and, (c4) the routing paths are sufficiently 
``spread-out.'' In practical networks however, these 
conditions may be violated, and as a consequence (P1) may perform poorly. For instance, if 
a bursty anomaly occurs, (c1) does not hold. A particular OD flow may also be 
inaccessible to sample via NetFlow, that violates (c3). Apparently, in the latter case, knowing the cross-correlation of a missing OD flow with other flows enables accurate interpolation of misses.

Inherent patterns of the nominal traffic matrix $\bX$ and the
anomalous traffic matrix $\bA$ can be learned from historical/training
data $\{\bx_t, \ba_t\}_{t \in \mathcal{H}}$, where $\bx_t$ and $\ba_t$ denote the network-wide nominal and anomalous traffic vectors at time $t$. Given the training data~$\{\bx_t, \ba_t\}_{t \in \mathcal{H}}$, link counts $\bY$ obeying~\eqref{eq:Y} as well as the partial flow-counts $\bZ_{\Pi}$ adhering to~\eqref{eq:Z}, and with $\{\bR,\Pi\}$ known, the rest of this paper deals with estimating the matrix pair~$(\bX,\bA)$.

\subsection{Bilinear Factorization}
\label{subsec:bilinear_fact}
The first step toward incorporating correlation information is to use the bilinear characterization of the nuclear norm. Using singular value decomposition~\cite{Horn}, one can always factorize the low-rank component as $\bX=\bL\bQ'$, where $\bL \in \mathbbm{R}^{F \times \rho}$, $\bQ \in \mathbbm{R}^{T \times \rho}$, for some $\rho \geq \rm{rank}(\bX)$. The nuclear-norm can then be redefined as (see e.g.,~\cite{srebro_2005})
\begin{align}
\|\bX\|_{\ast} := \min_{\bX=\bL {\bQ}'} \frac{1}{2} \{\|\bL\|_F^2 + \|\bQ\|_F^2 \}. \label{eq:nuc_norm_alt_def}
\end{align}
For the scalar case, \eqref{eq:nuc_norm_alt_def} leads to the identity 
$|a|=\min_{a=bc} \frac{1}{2} (|b|^2 + |c|^2)$. The latter implies that the 
$\ell_1$-norm of $\bA$ can be alternatively defined as
\begin{align}
\|\bA\|_{1} := \min_{\bA=\bB\odot\bC} \frac{1}{2} \{\|\bB\|_F^2 + \|\bC\|_F^2\}  \label{eq:l1_norm_alt_def}
\end{align}
where $\bB,\bC \in \mathbbm{R}^{F \times T}$. For notational convenience, let $\bU:=[\bY',~\bZ_{\Pi}^{'}]$ and the corresponding linear operator $\cP(\bX):=[(\bR\bX)',~\cP_{\Omega}(\bX)']$. Leveraging \eqref{eq:nuc_norm_alt_def} and \eqref{eq:l1_norm_alt_def}, one is prompted to recast (P1) as
\begin{align}
\text{(P4)}~~~&\min_{\{\bL,\bQ,\bB,\bC\}} ~~ \frac{1}{2}\|\bU-\cP(\bL\bQ'+\bB 
\odot \bC)\|_F^2 \nonumber \\ 
&+ \frac{\lambda_{\ast}}{2} \big\{\|\bL\|_F^2 + 
\|\bQ\|_F^2\big\} + \frac{\lambda_1}{2} \big\{\|\bB\|_F^2 + \|\bC\|_F^2\big\}. \nonumber
\end{align}
This Frobenius-norm regularization doubles the number of optimization 
variables for the sparse component~$\bA$ ($2FT$), but reduces the variable count for the 
low-rank component~$\bX$ to $\rho (F+T)$. Regarding performance, the bilinear factorization incurs no 
loss of optimality as stated in the next lemma.  

\begin{lemma}\label{lem:lemma_3}
If $\hat{\bX}$ denotes the optimal low-rank solution of (P1) and $\rho \geq 
\rm{rank}(\hat{\bX})$,~then (P4) is equivalent to (P1).
\end{lemma}

\begin{proof}
It readily follows from \eqref{eq:nuc_norm_alt_def} and \eqref{eq:l1_norm_alt_def} along with the commutative property of minimization which allows taking minimization first with respect to (w.r.t.) $\{\bL,\bQ\}$ and then w.r.t. $\{\bB,\bC\}$. 
\end{proof}

\section{Bayesian Traffic and Anomaly Estimates}
\label{sec:bayes_est}
This section recasts (P4) in a Bayesian framework by adopting the AWGN model 
$\bU=\cP(\bX+\bA) + \bE$, where $\bE$ contains independent identically 
distributed (i.i.d.) entries drawn from $\cN(0,\sigma^2)$. As in \eqref{eq:nuc_norm_alt_def} $\bX$ is also factorized as $\bL\bQ'$ with the independent factors 
$\bL:=[\bl_1,\ldots,\bl_{\rho}]$ and $\bQ:=[\bq_1,\ldots,\bq_{\rho}]$. Matrices $\bL$ 
and $\bQ$ are formed by i.i.d. columns obeying $\bl_i \sim \cN(0,\bR_L)$ and 
$\bq_i \sim \cN(0,\bR_Q)$, respectively, for positive-definite correlation 
matrices $\bR_L \in \mathbbm{R}^{F \times F}$ and $\bR_Q \in \mathbbm{R}^{T 
\times T}$. Without loss of generality (w.l.o.g.), in order to avoid the scalar ambiguity in $\bX=\bL\bQ'$ set 
$\rm{tr}(\bR_L)=\rm{tr}(\bR_Q)$. Likewise, the anomaly matrix is factored 
as $\bA=\bB \odot \bC$ with the independent factors $\bb:={\rm vec}(\bB) \in 
\mathbbm{R}^{FT}$ and $\bc :={\rm vec}(\bC)\in \mathbbm{R}^{FT}$ drawn from 
$\bb \sim \cN(0,\bR_B)$ and $\bc \sim \cN(0,\bR_C)$, with positive-definite correlation matrices $\bR_B,\bR_C \in \mathbbm{R}^{FT \times FT}$, respectively. 

For the considered AWGN model with priors, the maximum a posteriori (MAP) estimator of $(\bX,\bA)$ is given by the solution of
\begin{align}
&\text{(P5)}~~\min_{\{\bL,\bQ,\bB,\bC\}} \hspace{0mm} \frac{1}{2}\|\bU-\cP(\bL\bQ'+\bB \odot \bC)\|_F^2 \nonumber \\ &+ \frac{\lambda_1}{2} \big[\bb'\bR_B^{-1}\bb + \bc' \bR_C^{-1}\bc \big] + \frac{\lambda_{\ast}}{2} \big[{\rm tr}(\bL'\bR_L^{-1}\bL) + {\rm tr}(\bQ'\bR_Q^{-1}\bQ)\big]   \nonumber
\end{align}
for $\lambda_1=\lambda_{\ast}=\sigma^2$, where different weights $\lambda_1$ and $\lambda_{\ast}$ are considered here for generality. Observe that (P5) specializes to (P4) upon choosing $\bR_L= 
\bI_F$, $\bR_Q= \bI_T$, and $\bR_B=\bR_C= \bI_{FT}$. Lemma~\ref{lem:lemma_3} then implies that the convex program (P1) yields the MAP optimal estimator for the 
considered statistical model so long as the factors contain i.i.d. Gaussian entries. With respect to the 
statistical model for the low-rank and sparse components, as it will become clear 
later on, $\bR_L$ ($\bR_Q$) captures the correlation among columns (rows) of $\bX$;~likewise, $\bR_B$ and $\bR_C$ capture the correlation among entries of $\bA$.



Albeit clear in this section statistical formulation, the adopted model $\bX=\bL\bQ'$ promotes low rank as a result of ${\rm rank(\bX)} \leq \rho$, but it is not obvious whether $\bA=\bB \odot \bC$ effects sparsity. The latter will rely on the fact that the product of two independent Gaussian random variables is heavy tailed. To recognize this, consider the independent scalar random variables $b \sim \cN(0,1)$ and $c \sim \cN(0,1)$. The product random variable $a=bc$ can then be expressed as $bc=\frac{1}{4}(b+c)^2 - \frac{1}{4}(b-c)^2$, where $S_1:=\frac{1}{4}(b+c)^2$ and $S_2:=\frac{1}{4}(b-c)^2$ are central $\chi^2$-distributed random variables. Since $\mathbbm{E}[(a-b)(a+b)]=0$, the random variables $S_1$ and $S_2$ are independent, and consequently the characteristic function of $a$ admits the simple form $\Phi_a(\omega)=\Phi_{S_1}(\omega)\Phi_{S_2}(\omega)=1/(\sqrt{1+4\omega^2})$. Applying the inverse Fourier transform to $\Phi_a(\omega)$, yields the probability density function $p_a(x)=(1/\sqrt{2\pi}) k_0(x/2)$, where $k_0(x):=\int_{0}^{\infty}[\cos(\omega x)]/(\sqrt{1+4\omega^2}) ~d\omega$ denotes the modified Bessel function of second-kind, which is tightly approximated with $\sqrt{\pi/(2x)}e^{-x}$ for $x >1$~\cite[p. 20]{samko93}. One can then readily deduce that $p_a(x)=\sqrt{\pi/(2x)}e^{-|x|}$ behaves similar to the Laplacian distribution, which is well known to promote sparsity. In contrast with the Laplacian distribution however, the product of Gaussian random variables incurs a slightly lighter tail as depicted in Fig.~\ref{fig:fig_sparsity}. It is worth commenting that the correlated multivariate Laplacian distribution is an alternative prior distribution to postulate for the sparse component. However, its complicated form~\cite{laplacian_correlated} renders the optimization for the MAP estimator intractable.

\begin{figure}
\centering
\begin{tabular}{c}
     \epsfig{file=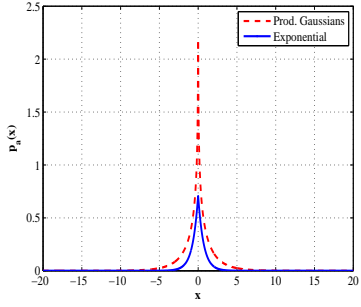,width=0.6\linewidth,height=1.7 in}  \\
  \end{tabular}
  \caption{Sparsity promoting priors with zero mean and unity variance. }
  \label{fig:fig_sparsity}
\end{figure}

\begin{remark}[nonzero mean] \label{remark:rem_nonzero_mean}
In general, one can allow nonzero mean for the factors in the adopted statistical model, and subsequently replaces correlations with covariances. This can be useful e.g., to estimate the nominal traffic which is inherently positive valued. The mean values are assumed zero here for simplicity.
\end{remark}


\subsection{Learning the correlation matrices}
\label{subsec:learning_covariance}
Implementing (P5) requires first obtaining the correlation matrices $\{\bR_L,\bR_Q,\bR_B,\bR_C\}$ from the second-order statistics of $(\bX,\bA)$, or their estimates based on training data. Given second-order statistics of the unknown nominal-traffic matrix $\bX$, matrices $\{\bR_L,\bR_Q\}$ can be readily found as explained in the next lemma. The proof is along the lines of~\cite{juan_tensor_tsp2013}, hence it is omitted for brevity.

\begin{lemma} \label{lem:lemma_4}
Under the Gaussian bilinear model for $\bX$, and with $\rm{tr}(\bR_L)=\rm{tr}(\bR_Q)$, it holds that
\begin{align}
&\bR_Q=\rho \mathbbm{E}[\bX'\bX]  / (\mathbbm{E}[\|\bX\|_F^2])^{1/2} , \nonumber \\
&\bR_L=\rho \mathbbm{E}[\bX \bX']  / (\mathbbm{E}[\|\bX\|_F^2])^{1/2} . \nonumber 
\end{align}
\end{lemma}
\hspace{-3.5mm} It is evident that $\bR_L$ captures {\it temporal} correlation of the network traffic (columns of $\bX$), while $\bR_Q$ captures the {\it spatial} correlation across OD flows (rows of $\bX$).

For real data where the distribution of unknowns is not available, $\{\bR_L,\bR_Q\}$ are typically estimated from the training data, which can be e.g., past estimates of nominal and anomalous traffic. For instance, consider $\{\bR_L,\bR_Q\}$ estimates as input to (P5) for estimating the traffic at day $K+1$ (corresponding to time horizon $\mathcal{T}$) with $T$ time instants, from the training data $\{\bx_t\}_{t=1}^{K T}$ collected during the past $K$ days. Apparently, reliable correlation estimates cannot be formed for general nonstationary processes. Empirical analysis of Internet traffic suggests adopting the following assumptions~\cite{LPC04}: (a1) Process $\{\bx_t\}$ is cyclostationary with a day-long period due to large-scale periodic trends in the nominal traffic; and (a2) OD flows are uncorrelated as their origins are mutually unrelated. One can also take into account weekly or monthly periodicity of traffic usage to further improve the accuracy of the correlation estimates.

Let $r_t$ denote the remainder of dividing $t$ by $T$. For time slots $t_1,t_2 \in \mathcal{T}$, (a1) asserts that the vector subprocesses $\{\bx_{kT+r_{t_1}}\}_{k=0}^{K-1}$ and $\{\bx_{kT+r_{t_2}}\}_{k=0}^{K-1}$ are stationary, and thus one can consistently estimate $\mathbbm{E}[\bx_{r_{t_1}}'\bx_{r_{t_2}}]$, to obtain $\bR_Q$ via the sample correlation $\frac{1}{K} \sum_{k=0}^{K-1} \bx_{kT+r_{t_1}}\bx_{kT+r_{t_2}}'$~\cite{gg_cyclostationary}. Likewise, the normalization term $\mathbbm{E}[\|\bX\|_F^2]$ is estimated relying on (a1) as $ \frac{1}{K} \sum_{t=1}^{T} \sum_{k=0}^{K-1} \|\bx_{kT+t}\|^2$. Estimating $\bR_L$ on the other hand relies on (a2). Let $\bxi_f^{'} \in \mathbbm{R}^T$ denote the time-series of traffic associated with OD flow $f$, namely the $f$-th row of $\bX$. It then follows from (a2) that $\mathbbm{E}[\bxi_{f_1}\bxi_{f_2}]=(\mathbbm{E}[\bxi_{f_1}])'(\mathbbm{E}[\bxi_{f_2}])$ for $f_1 \neq f_2 \in \mathcal{F}$, where due to (a1), $\mathbbm{E}[\xi_{f,t}]$ ($\xi_{f,t}$ signifies the $t$-th entry of $\bxi_f$) is estimated via the sample mean $\frac{1}{K} \sum_{k=0}^{K-1}x_{f,kT+r_t}$. Moreover, for $f_1=f_2=f$, the estimate for $\mathbbm{E}[\bxi_f'\bxi_{f}]$ is $\frac{1}{K} \sum_{k=0}^{K-1} \sum_{t=1}^T \xi_{f,kT+r_t}^2$.


Given the second-order statistics of $\bA$, the correlation matrices $\bR_B$ and $\bR_C$ are obtained next.

\begin{lemma} \label{lem:lemma_5}
Under the Gaussian bilinear model for $\ba={\rm vec}(\bA')$, it holds that $\mathbbm{E}[\ba\ba']=\bR_B \odot \bR_C$. 
\end{lemma}

\noindent In order to avoid the scalar ambiguity present in $\bR_B$ and $\bR_C$, assume equal-magnitude entries $|[\bR_B]_{i,j}|=|[\bR_C]_{i,j}|=|\big[\mathbbm{E}[\ba\ba']\big]_{i,j}|^{1/2},~\forall (i,j)$. Apparently, for a diagonal correlation matrix $\mathbbm{E}[\ba\ba']$, the factors are uniquely determined as $[\bR_B]_{i,i}=[\bR_C]_{i,i}=\big[\mathbbm{E}[\ba\ba']\big]_{i,i}^{1/2},~\forall i$. However, when nonzero off-diagonals are present, there may exist a sign ambiguity, and the signs should be assigned appropriately to guarantee that $\bR_B$ and $\bR_C$ are positive definite.

Correlation matrices $\{\bR_B,\bR_C\}$ required to run (P5) over the time horizon $\mathcal{T}$ ($|\mathcal{T}|=T$) are estimated from the training data $\{\ba_t\}_{t=1}^{KT}$ collected e.g., over the past $K$ days. Due to the diverse nature of anomalies, developing a universal methodology to learn $\bR_B$ and $\bR_C$ is an ambitious objective. Depending on the nature of anomalies, the learning process is possible under certain assumptions. One such reasonable assumption is that anomalies of different flows are uncorrelated, but for each OD flow, the anomalous traffic is stationary and possibly correlated over time. This model is appropriate e.g., when different flows are subject to bursty anomalies arising from unrelated external sources.



For the stationary anomaly process of flow $f$, namely $\{a_{f,t}\}_t$, let $R_a^{(f)}(\tau):=\mathbbm{E}[a_{f,t-\tau} a_{f,t}]$ denote the time-invariant cross-correlation. Let also $\balpha_f^{'}$ denote the $f$-th row of $\bA$, and introduce the correlation matrix $\bR_a^{(f)}:=\mathbbm{E}[\balpha_f \balpha_f^{'}] \in \mathbbm{R}^T$, which is Toeplitz with entries $[\bR_a^{(f)}]_{i,i+\tau}=R_a^{(f)} (\tau),~i\in [T],\tau =0,\ldots,T-1$. Accordingly, $\mathbbm{E}[\ba\ba']$ is a block-diagonal matrix with blocks $\bR_a^{(f)}$, and subsequently Lemma~\ref{lem:lemma_5} implies that $\bR_B$ and $\bR_C$ are block diagonal with Toeplitz blocks $\bR_b^{(f)}$ and $\bR_c^{(f)}$, respectively. Under the equal-magnitude assumption for the entries of $\bR_B$ and $\bR_c$, the entries of $\bR_b^{(f)}$ and $\bR_c^{(f)}$ are readily obtained as  
\begin{align}
&\big[\bR_b^{(f)}\big]_{i,i+\tau} \hspace{-2mm}= \big|R_a^{(f)}(\tau)\big|^{1/2}, \nonumber\\ &\big[\bR_c^{(f)}\big]_{i,i+\tau} \hspace{-2mm}= \big|R_a^{(f)}(\tau)\big|^{1/2} {\rm sgn}\big(R_a^{(f)}(\tau)\big).  \label{eq:R_B&R_C_stationary}
\end{align}
Notice that if $|R_a^{(f)}(\tau)|$ decays sufficiently fast as $\tau$ grows, $\bR_B$ and $\bR_C$ become positive definite~\cite{solo_adptfilter}. Finally, thanks to the stationarity of $\{a_{f,t}\}_t$, $R_a(\tau)$ can be consistently estimated using $\frac{1}{KT-\tau} \sum_{t=\tau+1}^{KT} a_{f,t-\tau} a_{f,t}$. It is worth noting that the considered model renders the sparsity regularizer in (P5) separable across rows of $\bA$, which in turn induces row-wise sparsity.

%



\section{Alternating Majorization-Minimization Algorithm}
\label{sec:alt_major_min}
In order to efficiently solve (P5), an alternating minimization (AM) scheme is developed here by alternating among four matrix variables $\{\bL,\bQ,\bB,\bC\}$.  The algorithm entails iterations updating one matrix variable at a time, while keeping the rest are kept fixed at their up-to-date values. In particular, iteration $k$ comprises orderly updates of four matrices $\bL[k] \rightarrow \bQ[k] \rightarrow \bB[k] \rightarrow \bC[k]$. For instance, $\bL[k]$ is updated given the latest updates $\{\bQ[k-1],\bB[k-1],\bC[k-1]\}$ as $\bL[k]=\arg\min_{\bL}g_L^{(k)}(\bL)$, where 
\begin{align}
g_L^{(k)}(\bL):= &
\frac{1}{2}\|\bU-\cP(\bL\bQ'[k-1]+\bB[k-1] \odot \bC[k-1])\|_F^2 \nonumber \\ & \hspace{2cm}+ 
\frac{\lambda_{\ast}}{2} \rm{tr}\big(\bL'\bR_{L}^{-1}\bL\big) \label{eq:g_l}
\end{align}
Likewise, $\bQ[k]$, $\bB[k]$, and $\bC[k]$ are updated by respectively minimizing $g_Q^{(k)},g_B^{(k)}$, and $g_C^{(k)}$, which are given similar to $g_L^{(k)}$ based on latest updates of the corresponding variables.

Functions $\{g_L^{(k)},g_Q^{(k)},g_B^{(k)},g_C^{(k)}\}$ are strongly convex quadratic programs due to regularization with positive definite correlations in the regularizer, and thus their solutions admits closed form after inverting certain possibly 
large-size matrices. For instance, updating $\bL[k]$ requires inverting an $F\rho 
\times F\rho$ matrix. This however may not be affordable since in practice 
the number of flows $F$ is typically $\mathcal{O}(N^2)$, which can be too large. To cope with this curse of dimensionality, instead of 
$\{g_{L}^{(k)},g_{Q}^{(k)},g_{B}^{(k)},g_{C}^{(k)}\}$ judicious surrogates 
$\{\tilde{g}_{L}^{(k)},\tilde{g}_{Q}^{(k)},\tilde{g}_{B}^{(k)},\tilde{g}_{C}^{(
k)}\}$, chosen based on the second-order Taylor-expansion around the previous 
updates, are minimized. As will be clear later, adopting these surrogates avoids inversion, and parallelizes the computations. The 
aforementioned surrogate for $g_{L}^{(k)}$ around $\bL[k-1]$ is given as
\begin{align}
\tilde{g}_{L}^{(k)}(\bL):= &g_{L}^{(k)}(\bL[k-1]) + \tr\big((\bL-\bL[k-1])'\nabla 
g_{L}^{(k)}(\bL[k-1]) \big)  \nonumber \\
&\hspace{2cm}+ \frac{\mu_{L}[k]}{2} \|\bL-\bL[k-1]\|_F^2 \label{eq:g_tilde}
\end{align}
for some $\mu_{L}[k] \geq \sigma_{\max}\big[\nabla^2 g_{L}^{(k)}(\bL[k-1])\big]$ (likewise for 
$\tilde{g}_{Q}^{(k)},\tilde{g}_{B}^{(k)}$, and $\tilde{g}_{C}^{(k)}$). It is useful 
to recognize that each surrogate, say $\tilde{g}_{L}^{(k)}$, has the following properties: (i) it majorizes $g_{L}^{(k)}$, namely $g_{L}^{(k)}(\bL) \leq \tilde{g}_{L}^{(k)}(\bL),~\forall \bL$; and it is locally tight, which means that (ii) $ 
g_{L}^{(k)}(\bL[k-1])= \tilde{g}_{L}^{(k)}(\bL[k-1])$; and, (iii) $\nabla 
g_{L}^{(k)}(\bL[k-1])=\nabla \tilde{g}_{L}^{(k)}(\bL[k-1])$.

The sought approximation leads to an iterative procedure, where iteration $k$ entails orderly updating $\{\bL[k],\bQ[k],\bB[k],\bC[k]\}$ by minimizing $\tilde{g}_{L}^{(k)},\tilde{g}_{Q}^{(k)},\tilde{g}_{B}^{(k)},\tilde{g}_{C}^{(k)}$, respectively; e.g., the update for $\bL[k]$ is  
\begin{align}
\bL[k] & =  \arg\min_{\bL \in \mathbbm{R}^{F \times \rho}} \tilde{g}_{L}^{(k)}(\bL) \nonumber \\ 
&=\hspace{-1mm} \bL[k-1] - (\mu_{L}[k])^{-1} \nabla g_{L}^{(k)}(\bL[k-1]) \nonumber
\end{align}
which is a nothing but a single step of gradient descent on $g_{L}^{(k)}$. Upon defining the residual matrices 
$\bdelta_y(\bL,\bQ,\bB,\bC):=\bR(\bL\bQ'+\bB\odot\bC)-\bY$ and 
$\bdelta_z(\bL,\bQ,\bB,\bC):=\cP_{\Pi}(\bL\bQ'+\bB\odot\bC)-\bZ_{\Pi}$, the overall algorithm is listed in Table~\ref{tab:alt_maj_min}.

All in all, Algorithm~\ref{tab:alt_maj_min} amounts to an iterative block-coordinate-descent scheme with four block updates per iteration, each minimizing a tight surrogate of (P5). Since each subproblem is smooth and strongly convex, the convergence follows from~\cite{meisam_bsum_2013} as stated next.

\begin{proposition}\cite{meisam_bsum_2013}
Upon choosing $\{c_L^{'} \geq \mu_L[k] \geq \sigma_{\max}\big[\nabla^2 g_{L}^{(k)}(\bL[k-1])\big]\}_{k=1}^{\infty}$ for some $c_L^{'} > 0$ (likewise for $\mu_Q[k],\mu_B[k],\mu_C[k]$), the iterates $\{\bL[k],\bQ[k],\bB[k],\bC[k]\}$ generated by Algorithm~\ref{tab:alt_maj_min} converge to a stationary point of (P5).
\end{proposition}

\begin{remark}[Fast algorithms]In order to speed up the gradient descent iterations per block of Algorithm~\ref{tab:alt_maj_min}, Nesterov-type acceleration techniques along the lines of those introduced in e.g.,~\cite{nesterov83} can be deployed, which can improve the $\mathcal{O}(1/k)$ convergence rate of the standard gradient descent to $\mathcal{O}(1/k^2)$.
\end{remark}

\begin{algorithm}[t]
\caption{: Alternating majorization-minimization solver for (P5)} \small{
\begin{algorithmic}
    \STATE \textbf{input} $\bY, \bZ_{\Pi}, \Pi, \bR, 
    \bR_{L},\bR_Q,\bR_B,\bR_C, \lambda_{\ast}, \lambda_1,$ \STATE \hspace{0.75cm} and $\{\mu_L[k],\mu_Q[k],\mu_B[k],\mu_C[k]\}_{k=1}^{\infty}.$
    \STATE \textbf{initialize} $\bL[0],\bQ[0],\bB[0],\bC[0]$ at random, and set $k=0$.
    \WHILE {not converged}

        \STATE \textbf{[S1] Update $\bL$}
\STATE $\bF[k] = \bR'\bdelta_y(\bL[k],\bQ[k],\bB[k],\bC[k]) + 
        \bdelta_z(\bL[k],\bQ[k],\bB[k],\bC[k]) $
        \STATE $\bL[k+1] = \bL[k] - \frac{1}{\mu_{L}[k]} \big( \bF[k]\bQ[k] + \lambda_{\ast} \bR_L^{-1} \bL[k]  \big)$

        \STATE \textbf{[S2] Update $\bQ$}
 \STATE $\bG[k]= 
       \bdelta_y'(\bL[k+1],\bQ[k],\bB[k],\bC[k])\bR + 
       \bdelta_z'(\bL[k+1],\bQ[k],\bB[k],\bC[k]) $
       \STATE $\bQ[k+1] = \bQ[k] - \frac{1}{\mu_{Q}[k]}  \Big[ \bG[k] \bL[k+1] + 
       \lambda_{\ast} \bR_{Q}^{-1} \bQ[k] \Big] $

        \STATE \textbf{[S3] Update $\bB$}
        
        \STATE $\bH[k]= \bR'\bdelta_y(\bL[k+1],\bQ[k+1],\bB[k],\bC[k])+ \bdelta_z(\bL[k+1],\bQ[k+1],\bB[k],\bC[k])  $

        \STATE $\bB[k+1] = \bB[k] - \frac{1}{\mu_{B}[k]}  \Big[ \bC[k] \odot \bH[k] + \lambda_{1} \rm{unvec}\big(\bR_B^{-1}\rm{vec}(\bB[k])\big) \Big] $
        
        \STATE \textbf{[S4] Update $\bC$}

\STATE $\bE[k]=  \bR'\bdelta_y(\bL[k+1],\bQ[k+1],\bB[k+1],\bC[k]) + 
        \bdelta_y(\bL[k+1],\bQ[k+1],\bB[k+1],\bC[k])  $
        \STATE $\bC[k+1] = \bC[k] - \frac{1}{\mu_{C}[k]} \Big[ \bB[k] \odot \bE[k]
         + \lambda_1 
        \rm{unvec}\big(\bR_C^{-1}\rm{vec}(\bC[k])\big) \Big]$

        \STATE $k\leftarrow k+1$
    \ENDWHILE
    \RETURN $(\bA[k]=\bB[k] \odot \bC[k],\bX[k]=\bL[k]\bQ'[k])$
\end{algorithmic}}
\label{tab:alt_maj_min}
\end{algorithm}

\section{Practical Considerations}
\label{sec:pr_issues}
Before assessing their relevance to large-scale networks, the proposed algorithms must address additional practical issues. Those relate to the fact that network data are typically decentralized, streaming, subject to outliers as well as misses, and the routing matrix may be either unknown or dynamically changing over time. This section sheds light on solutions to cope with such practical challenges.

\subsection{Inconsistent partial measurements}
\label{subsec:inconsistent}
Certain network links may not be easily accessible to collect measurements, or, their measurements might be lost during the communication process due to e.g., packet drops. Let $\Pi_y$ collect the available link measurements during the time horizon $\cT$. In addition, certain link or flow counts may not be consistent with the adopted model in \eqref{eq:Y} and \eqref{eq:Z}. To account for possible presence of outliers introduce the matrices $\bO_y \in \mathbbm{R}^{L \times T}$ and $\bO_z \in \mathbbm{R}^{F \times T}$, which are nonzero at the positions associated with the outlying measurements, and zero elsewhere. The link-count model \eqref{eq:Y} should then be modified to $\bY_{\Pi_y}=\cP_{\Pi_y}(\bR(\bX+\bA)+\bO_y+\bV)$, and the flow counts to $\bZ_{\Pi}=\cP_{\Pi}(\bX+\bA+\bO_z+\bW)$. Typically the outliers constitute a small fraction of measurements, thus rendering $\{\bO_y,\bO_z\}$ sparse. The optimization task (P1) can then be modified to take into account the misses and outliers as follows
\begin{align}
&\text{(P6)}~(\hat{\bX},\hat{\bA})\hspace{-1mm}=\hspace{-1mm}\arg\hspace{-3mm}\min_{\{\bX,\bA,\bO_y,\bO_z\}} \frac{1}{2}\|\cP_{\Pi_y}(\bY - \bR(\bX+\bA)-\bO_y)\|_F^2 \nonumber \\ 
&\hspace{1cm}+ \frac{1}{2}\|\cP_{\Pi}(\bZ-\bX-\bA - \bO_z)\|_F^2 + \lambda_{\ast}\|\bX\|_{\ast} + \lambda_1\|\bA\|_1 \nonumber \\
&\hspace{4cm}+ \lambda_{y} \|\bO_y\|_1 + \lambda_{z} \|\bO_z\|_1
\nonumber
\end{align}
where $\lambda_y$ and $\lambda_z$ control the density of link- and flow-level outliers, respectively. Again, one can employ ADMM-type algorithms to solve (P6).

Routing information may not also be revealed in certain applications due to e.g., privacy reasons. In this case, each network link can potentially carry an unknown fraction of every OD flow. Let $\cL_{\rm in}(n)$ and $\cL_{\rm out}(n)$ denote the set of incoming and outgoing links to node $n \in \cN$. The routing variables then must respect the flow conservation constraints, that is formally~$\bR \in \cR:=\{\bR \in [0,1]^{L \times F}~:~ \sum_{\ell \in \mathcal{L}_{\rm in}(n)} r_{\ell,f} = \sum_{\ell \in 
\mathcal{L}_{\rm out}(n)} r_{\ell,f},~ \forall f \in \mathcal{F},n \in \cN\}$. Taking the unknown routing variables into account, the optimization task to estimate the traffic is formulated as
\begin{align}
&\text{(P7)}~~(\hat{\bX},\hat{\bA})=\arg\min_{\{\bX,\bA, \bR \in \cR \}}  \frac{1}{2}\|\bY - \bR(\bX+\bA)\|_F^2 \nonumber \\ 
&\hspace{1.5cm}+ \frac{1}{2}\|\cP_{\Pi}(\bZ-\bX-\bA )\|_F^2 + \lambda_{\ast}\|\bX\|_{\ast} + \lambda_1\|\bA\|_1
\nonumber
\end{align}
which is nonconvex due to the presence of bilinear terms in the LS cost.

\subsection{Real-time operation}
\label{subsec:real_time}
Monitoring of large-scale IP networks necessitates collecting
massive amounts of data which far outweigh the ability of
modern computers to store and analyze them in real time. In addition, nonstationarities due to routing changes and missing
data further challenges estimating traffic and anomalies. In dynamic
networks routing tables are constantly readjusted to effect
traffic load balancing and avoid congestion caused by e.g.,
traffic congestion anomalies or network infrastructure failures. On top of the previous arguments, in practice the measurements are acquired sequentially across time, which motivates updating previously obtained estimates rather than recomputing new ones from scratch each time a new datum becomes available.

To account for routing changes, let $\bR_t \in \mathbbm{R}^{L \times F}$ denote the routing matrix at time $t$. The observed link counts at time instant $t$ then adhere to $\by_t=\bR_t(\bx_t+\ba_t)+\bv_t,~t=1,2,\ldots$, where $\by_t \in \mathbbm{R}^L$, and the partial flow counts at time~$t$ obey $\bz_{\Pi_t}=\cP_{\Pi_t}(\bx_t + \ba_t + \bw_t),~t=1,2,\ldots$, where $\bz_{\Pi_t} \in \mathbbm{R}^{F}$, and $\Pi_t$ indexes the OD flows measured at time $t$. In order to estimate the nominal and anomalous traffic components $(\bx_t,\ba_t)$ at time instant $t$ in real time, given only the past observations $\{\by_{\tau},\bz_{\Pi_{\tau}}\}_{\tau=1}^t$, the framework developed in our companion paper~\cite{jstsp_anomalography_2012} can be adopted. Building on the fact that the traffic traces $\{\bx_t\}_{t=1}^{\infty}$ lie in a low-dimensional linear subspace, say $\cL$, one can postulate $\bx_t=\bL\bq_t$ for $\bL \in \mathbbm{R}^{F \times \rho}$ with $\rho \ll F$, where $\bL$ spans the subspace $\cL$. Pursuing the ideas in \cite{jstsp_anomalography_2012}, the nuclear-norm characterization in \eqref{eq:nuc_norm_alt_def}, which enjoys separability across time, can be applied to formulate exponentially-weighted LS estimators. The corresponding optimization task can then be solved via alternating minimization algorithms~\cite{jstsp_anomalography_2012}.

It is worth commenting that the companion work~\cite{jstsp_anomalography_2012} aims primarily at identifying the anomalies $\ba_t$ from link counts, which requires slow variations of the routing matrix to ensure $\{\bR_t\bx_t\}_{t=1}^{\infty}$ lie in a low-dimensional subspace. However, the tomography task considered in the present paper imposes no restriction on the routing matrix. Indeed, routing variability helps estimation of the nominal traffic $\bx_t$. More precisely, suppose that $\{\bR_{t}\}$ are sufficiently distinct so as the intersection of the nullspaces~$\bigcap_{t} \cN_{R_t}$ has a small dimension. Consequently, it is less likely to find an alternative feasible solution $\bX_1:=\bX_0+\bH$ with $\bH:=[\bh_1,\ldots,\bh_F]$ and $\bh_t \in \cN_{R_t}$ such that $\bH \in \Phi_{X_0}$~(cf. Section~\ref{sec:rec_guarantees}); see also Lemma~\ref{lemma:lem_1}. Further analysis of this intriguing phenomenon goes beyond the scope of the present paper, and will be pursued as future research.

\subsection{Decentralized implementation}
\label{subsec:dec_alg}
Algorithms~\ref{tab:admom_alg} and~\ref{tab:alt_maj_min} demand each network node (router) $n \in \cN$ continuously communicate the local measurements of its incident links as well as the OD-flow counts originating at node $n$, to a central monitoring station. While this is typically the prevailing operational paradigm adopted in current network technologies, there are limitations associated with this architecture. Collecting all these data at the routers may lead to excessive protocol overhead, especially for large-scale networks with high acquisition rate. In addition, with the exchange of raw measurements missing data due to communication errors are inevitable. Performing the optimization in a centralized fashion raises robustness concerns as well, since the central monitoring station represents an isolated point of failure.

The aforementioned reasons motivate devising fully distributed iterative algorithms in large-scale networks, which allocate the network tomography functionality to the routers. In a nutshell, per iteration, nodes carry out simple computational tasks locally, relying on their own local measurements. Subsequently, local estimates are refined after exchanging messages only with directly connected neighbors, which facilitates percolation of information to the entire network. The ultimate goal is for the network nodes to consent on the global map of network-traffic-state~$(\hat{\bX},\hat{\bA})$, which remains close to the one obtained via the centralized counterpart with the entire network data available at once. Building on the separable characterization of the nuclear norm in~\eqref{eq:nuc_norm_alt_def}, and adopting ADMM method as a basic tool to carry out distributed optimization, a generic framework for decentralized sparsity-regularized rank minimization was put forth in our companion paper~\cite{tsp_rankminimization_2012}. In the context of network anomaly detection, the results there are encouraging and the proposed ideas can be applied to solve also (P1) in a distributed fashion.

%


\section{Performance Evaluation}
\label{sec:perf_eval}

Performance of the novel schemes is assessed in this section via computer simulations with both synthetic and real network data as described below.

\noindent\textbf{Synthetic network data.}~The  network topology is generated according to a random geometric graph model, where the nodes are randomly placed in a unit square, and two nodes are connected with an edge if their distance is less than a prescribed threshold $d_c$. In general, to form the routing matrix each OD pair takes $K$ nonoverlapping paths, each determined according to the minimum hop-count algorithm. After finding the routes, links carrying no traffic are discarded. Clearly, the number of links varies according to $d_c$. The underlying traffic matrix $\bX_0$ follows the bilinear model $\bX_0 = \bL\bQ'$, with the factors $\bL \in \mathbbm{R}^{F \times \rho}$ and $\bQ \in \mathbbm{R}^{T \times \rho}$ having i.i.d. Gaussian entries $\mathcal{N}(0,1/F)$ and $\mathcal{N}(0,1/T)$, respectively. Entries of the anomaly matrix $\bA_0$ are also randomly drawn from the
set $\{-1,0,1\}$ with probability (w.p.) ${\rm Pr} (a_{f,t}=-1)={\rm Pr}(a_{f,t}=1)=p/2$, and ${\rm Pr} (a_{f,t}=0)=1-p$. The link loads are then formed as $\bY=\bR(\bX_0+\bA_0)$. A subset of OD flows is also sampled uniformly at random to form the partial OD flow-level measurements $\bZ_{\Pi}=\mathbf{\Pi} \odot (\bX_0+\bA_0)$, where each entry of $\mathbf{\Pi} \in \{0,1\}^{F \times T}$ is i.i.d. Bernoulli distributed taking value one w.p. $\pi$, and zero w.p. $1-\pi$.



\noindent\textbf{Real network data}. Real data including OD flow
traffic levels are collected from the operation of the Internet-2 network (Internet backbone network across USA)~\cite{Internet2}, shown in Fig.~\ref{fig:net_topol} (a). Internet-2 comprises $N=11$ nodes,
$L=41$ links, and $F=121$ OD flows. Flow traffic
levels are recorded every five-minute interval, for a three-week operational period during December 8-28, 2003~\cite{Internet2,crovella_traffic_data}. The collected flow levels are the aggregation of clean and anomalous traffic components, that is sum of unknown ``ground-truth'' low-rank and sparse matrices $\bX_{0}+\bA_0$. The ``ground truth'' components are then discerned from their aggregate after applying robust PCP algorithms developed e.g., in~\cite{CLMW09}. The recovered $\bX_0$ exhibits three dominant singular values, confirming the low-rank
property of the nominal traffic matrix. Also, after retaining only the significant spikes with magnitude larger than the threshold $50\|\bY\|_F/LT$, the formed anomaly matrix $\bA_0$ has $1.10\%$ nonzero entries. The link loads in $\bY$ are obtained through multiplication of the aggregate traffic with the Internet-2 routing matrix. Even though $\bY$ is ``constructed'' here from flow measurements, link loads are acquired from SNMP traces~\cite{MC03}. Moreover, the aggregate flow traffic matrix $\bX_0+\bA_0$ is sampled uniformly at random with probability $\pi$ to form $\bZ_{\Pi}$. In practice, these samples are acquired via NetFlow protocol~\cite{robusttrafficestimation_Zhao06}.


\begin{figure}[t]
\centering
\begin{tabular}{cc}
\hspace{0mm}\epsfig{file=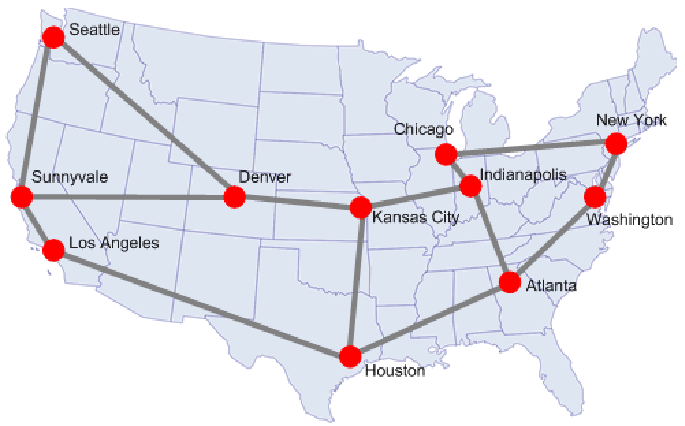,width=0.4
\linewidth, height=0.29
\linewidth } & \hspace{0mm}
\epsfig{file=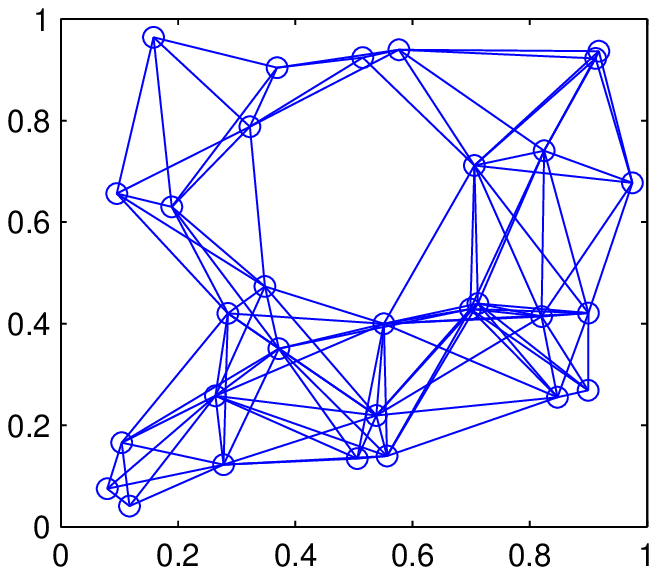,width=0.42
\linewidth, height=0.35
\linewidth } \\
(a) &
(b) \\
 \end{tabular}
 \caption{Network topology graphs. (a) Internet-2. (b) Random synthetic network with $N=30$ and $d_c=0.35$. }
 \label{fig:net_topol}
\end{figure}

\subsection{Exact recovery validation}
\label{subsec:exact_recovery}
To demonstrate the merits of (P2) in accurately recovering the true values $(\bX_0,\bA_0)$, it is solved for a wide range of rank $r$ and (average) sparsity levels $s=pFT$ using the ADMM solver in Algorithm~\ref{tab:admom_alg}. Synthetic data is generated as described before for a random network with $N=30$, $d_c=0.35$, and $F=T=N(N-1)/3$; see Fig.~\ref{fig:net_topol}(b). For $F$ randomly selected OD pairs, $K$ nonoverlapping paths are chosen to carry the traffic. Each path is created based on the minimum-hop count routing algorithm to form the routing matrix. A random fraction of the origin's traffic is also assigned to each path. The gray-scale plots in Fig.~\ref{fig:fig_perfsynthetic} show phase transition for the relative estimation error $e_{x+a}=e_x+e_a$, including both nominal $e_x:=\|\hat{\bX}-\bX_0\|_F/\|\bX_0\|_F$, and anomalous traffic estimation error $e_a:=\|\hat{\bA}-\bA_0\|_F/\|\bA_0\|_F$ under various percentage of misses. Top figure is associated with $K=1$, while for the bottom figure $K=3$. The parameter $\lambda$ in (P2) is also tuned to optimize the performance. 

When single-path routing is used, the network entails $L=159$ physical links. In this case, the routing matrix $\bR \in \{0,1\}^{159 \times 290}$ has a huge nullspace with ${\rm dim}(\cN_R)=127$, and as a result Fig.~\ref{fig:fig_perfsynthetic} (top) indicates that accurate recovery is possible only for relatively small values of $r$ and $s$. However, when multipath routing ($K=3$) is used, there are more $L=227$ physical links involved in carrying the traffic of OD flows. This shrinks the nullspace of $\bR \in [0,1]^{227 \times 290}$ to ${\rm dim}(\bR)=68$, and improves the isometry property of $\bR$ for sparse vectors. As a result, under traffic of higher dimensionality and denser anomalies accurate traffic estimation is possible; see Fig.~\ref{fig:fig_perfsynthetic} (bottom).

%
%
%
%

%

\begin{figure}[t]
\centering
\begin{tabular}{c}
\epsfig{file=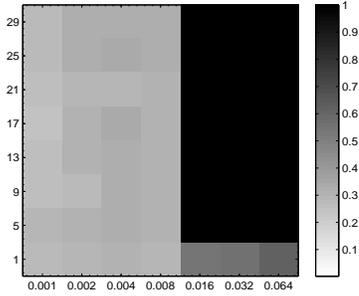,width=0.65
     \linewidth,
     height=0.5
     \linewidth } \\ \vspace{-4mm}
    
      \epsfig{file=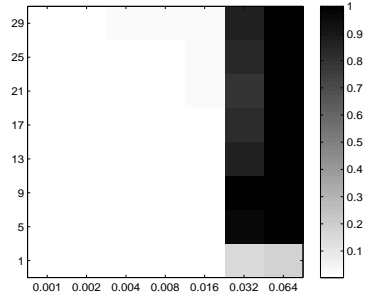,width=0.65
     \linewidth,
     height=0.5
     \linewidth } 
     
     \end{tabular}
     
  \caption{Relative estimation error $e_{x+a}$ for various values of rank ($r$) and sparsity level ($s=p F T$) where $F=T=290$ and $\pi=0.25$. (a) Single-path routing versus (b) multipath routing ($K=3$). White represents exact recovery ($e_{x+a}\thickapprox0$), while black represents $e_{x+a} \thickapprox 1$.} \label{fig:fig_perfsynthetic}
\end{figure}

%

\subsection{Traffic and anomaly maps}
\label{ssec:traffic_est}
Real Internet-2 data is considered to portray the traffic based on (P1) every $42$-hour interval, which amounts to time horizon of $T=504$ time bins.

\hspace{-4mm}\textbf{Impact of NetFlow data.}~The role of NetFlow measurements on the traffic estimation performance is depicted in Fig.~\ref{fig:fig_relerr_vs_p} plotting the relative error $e_{x+a}$ for various percentages of NetFlow samples~($\pi$). Normally, the estimation accuracy improves as $\pi$ grows, where the improvement seems more pronounced for the nominal traffic. When only the link loads are available, adding $10\%$ NetFlow samples enhances the nominal-traffic estimation accuracy by $45\%$, while the one for the anomalous traffic is improved by $18\%$. This observation corroborates the effectiveness of exploiting partial NetFlow samples toward mapping out the network traffic.

\hspace{-4mm}\textbf{Traffic profiles.}~For $\pi=0.1$, the true and estimated traffic time-series are illustrated in Fig.~\ref{fig:fig_perf_realdata} for three representative OD flows originating from the CHIN autonomous system located at Chicago. The depicted time-series correspond to three different rows of $\hat{\bX}$ and $\hat{\bA}$ returned by (P1). It is apparent that the traffic variations are closely tracked and significant spikes are correctly picked by (P1). It pinpoints confidently a significant anomaly occurring within 9:20~P.M.--9:25~P.M., December 11, 2003, in the flow CHIN--LOSA, which traverses several physical links. High false alarm declared for the CHIN--IPLS flow is also because it visits only a single link, and thus not revealing enough information.

\begin{figure}[t]
\centering
  \centerline{\epsfig{figure=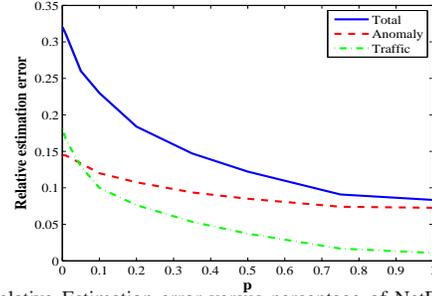,width=0.35\textwidth ,height=0.23\textwidth}}
\vspace{-5mm}\caption{Relative Estimation error versus percentage of NetFlow samples.}
  \label{fig:fig_relerr_vs_p}
\end{figure}

\begin{figure}[t]
\centering
\begin{tabular}{cc}
\hspace{-6mm}\epsfig{file=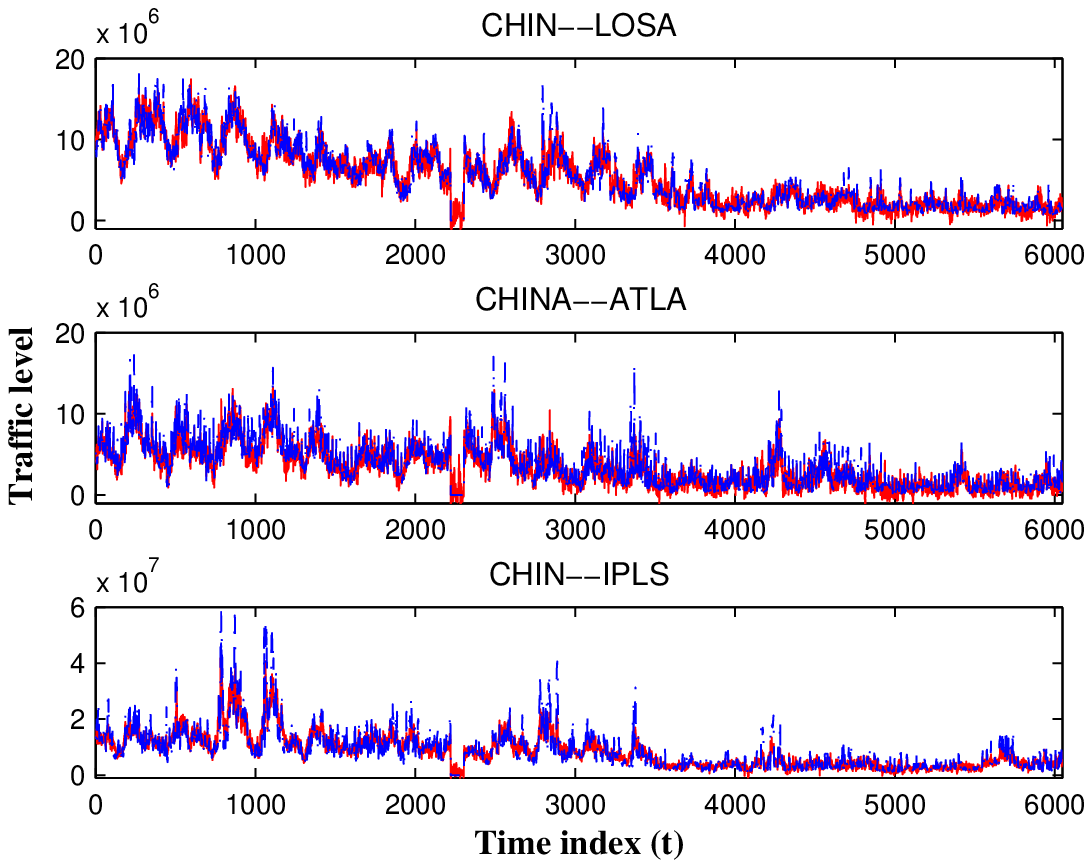,width=0.54
\linewidth, height=2.3 in } & \hspace{-8mm}
\epsfig{file=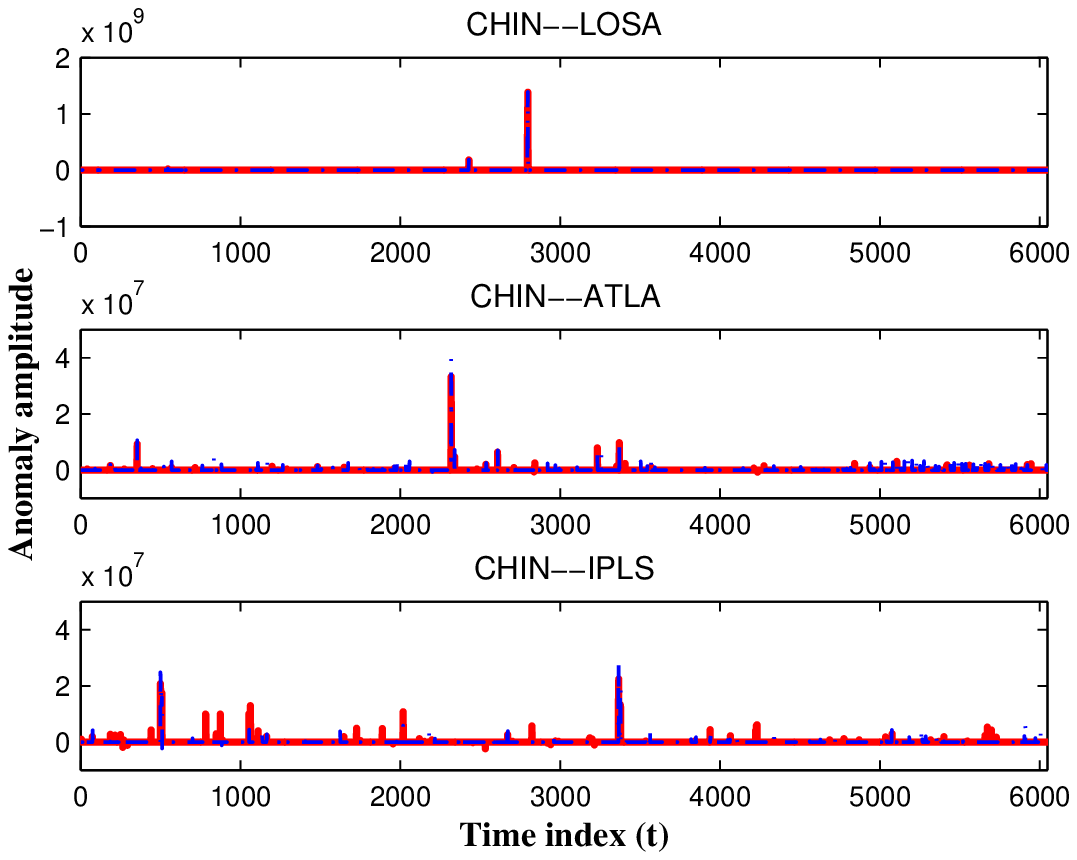,width=0.54
\linewidth, height=2.3 in } \\
(a) &
(b) \\
 \end{tabular}
 \caption{Nominal (a) and anomalous (b) traffic portrays for three representative OD flows when $\pi=0.1$. True traffic is dashed blue and the estimated one is solid red.}
 \label{fig:fig_perf_realdata}
\end{figure}

\hspace{-4mm}\textbf{Unveiling anomalies.}~Identifying anomalous patterns is pivotal towards proactive network security tasks. The resultant estimated map $\hat{\bA}$ returned by (P1) offers a depiction of the network health-state across both time and flows. Our previous work in~\cite{tit_exactrecovery_2012} and~\cite{jstsp_anomalography_2012} deals with creating such a map with only the link loads $\bY$ at hand (i.e.,~$\Pi = \emptyset$), and the primary goal is to recover $\hat{\bA}$. The purported results in~\cite{tit_exactrecovery_2012,jstsp_anomalography_2012} are promising and could markedly outperform state-of-art workhorse PCA-based approaches in e.g.,~\cite{LCD04,zggr05}. Relative to~\cite{tit_exactrecovery_2012,jstsp_anomalography_2012}, the current work however allows additional partial flow-level measurements. This naturally raises the question how effective this additional information is toward identifying the anomalies. As seen in Fig.~\ref{fig:fig_relerr_vs_p}, taking more NetFlow samples is useful, but beyond a certain threshold it does not offer any extra appeal.



\subsection{Estimation with spatiotemporal correlation information}
\label{subsec:cor_info}
This section evaluates the effectiveness of (P5) and demonstrates the usefulness of traffic correlation information. Training data from the week December 8-15, 2003 are used to estimate the Internet-2 traffic on the next day, December 16, 2003. The nominal ``ground truth'' traffic matrix $\bX_0$ described earlier is considered, and for validation purposes bursty anomalies are synthetically injected to form the aggregate traffic $\bX_0+\bA_0$, which is then used to generate $\bY$ and $\bZ_{\Pi}$. To simulate the NetFlow samples, suppose $10 \%$ of randomly selected OD flows are inaccessible for the entire time horizon, and the rest are sampled only $10 \%$ of time, resulting in $9 \%$ flow-level measurements available.

\hspace{-4mm}\textbf{Bursty anomalies.}~To generate anomalies~$\bX_0$, envision a scenario where a subset of OD flows undergo bursty anomalies while the rest are clean. Per flow $f$ bursty anomalies are generated according to the random multiplicative process $\{a_{f,t}=\gamma_f b_{f,t} c_{f,t} \}_t$, with mutually independent stationary processes $\{c_{f,t}\}$ and~$\{b_{f,t}\}$. The former is a correlated Gaussian process, and the latter is a correlated $\{0,1\}$-Bernoulli process to model the bursts. The Gaussian process obeys the first-order auto-regressive model $c_{f,t}=\theta  c_{f,t-1}+ \sigma_n n_{f,t}$, with $c_{f,0}=0$ and $n_{f,t} \sim \cN(0,1)$ for some $\theta <1$. The Bernoulli process also adheres to $b_{f,t}= d_{f,t} b_{f,t-1}+(1-d_{f,t})e_{f,t}$, where the independent random variables $d_{f,t}$ and $e_{f,t}$ obey $d_{f,t} \sim {\rm Ber} (\alpha)$ and $e_{f,t} \sim {\rm Ber}(\nu)$, respectively. Initial variable $b_{f,0}$ is also generated as ${\rm Ber}(\nu)$.

\hspace{-4mm}\textbf{Learning correlations.}~Owing to the stationarity of processes $\{b_{f,t}\}$ and $\{c_{f,t}\}$, process $\{a_{f,t}\}$ is stationary, and as a result $R_a^{(f)}(\tau)=\gamma_f^2 R_b^{(f)}(\tau) R_c^{(f)}(\tau)$, with the corresponding correlations given as~$R_c^{(f)}(\tau)=\theta^{\tau} \sigma_n^2/(1-\theta^2)$ and $R_b^{(f)}(\tau)=\nu (1-\nu) \alpha^{\tau} + \nu$. Set $\gamma_f=50$, $\theta=0.999$, $\sigma_n=0.005$, $\alpha=0.98$, and $\nu=0.03$. The correlation matrices $\{\bR_B,\bR_C\}$ with Toeplitz blocks are then obtained from~\eqref{eq:R_B&R_C_stationary}. Moreover, to account for the cyclostationarity of traffic with a day-long periodicity, the correlation matrices $\{\bR_L,\bR_Q\}$ are learned as elaborated in Section~\ref{subsec:learning_covariance}. The resulting temporal correlation matrices $\bR_B$ and $\bR_Q$, learned based on the traffic data December 8-15, 2003, are displayed in Fig.~\ref{fig:fig_correlation_RB_RQ}, where $288$ data points in each axis correspond to $24$ hours. The sharp transition noticed in Fig.~\ref{fig:fig_correlation_RB_RQ} (b) happens at $3:45$ p.m. that signifies a sudden increase in the traffic usage for the rest of the day.

\begin{figure}[t]
\centering
\begin{tabular}{cc}
\hspace{-2mm}\epsfig{file=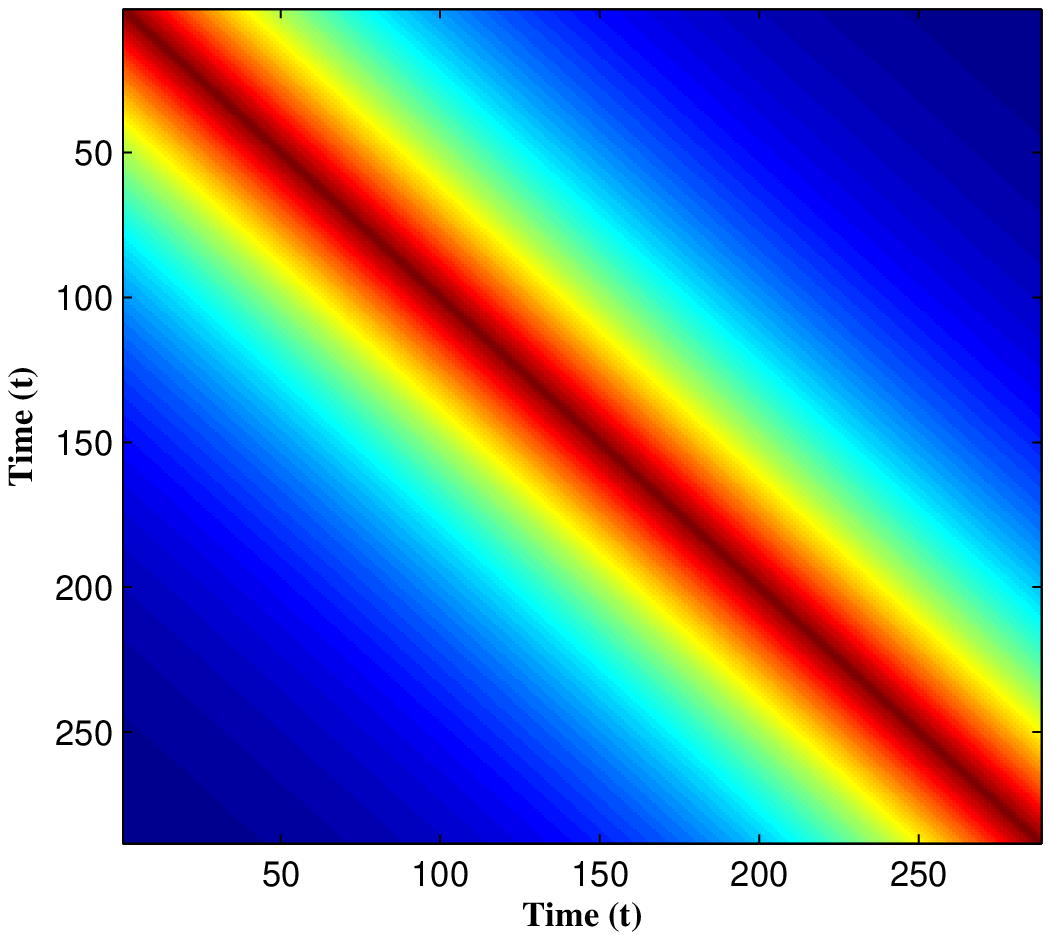,width=0.43
\linewidth, height=0.38 \linewidth } & \hspace{-3mm}
\epsfig{file=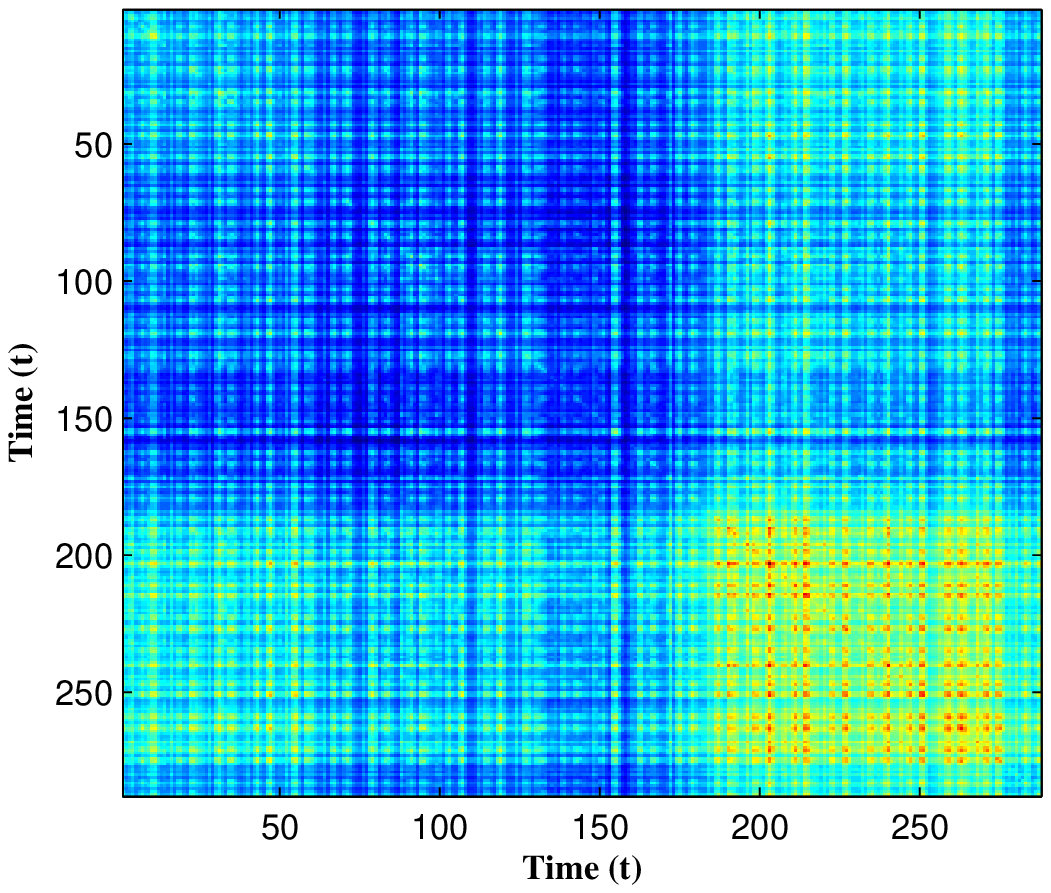,width=0.43
\linewidth, height=0.38 \linewidth } \\
(a) &
(b) \\
 \end{tabular}
 \caption{Sample correlations $\bR_B$ (a) and $\bR_Q$ (b) learned based on historical traffic data during December 8-15, 2003.}
 \label{fig:fig_correlation_RB_RQ}
\end{figure}

\hspace{-4mm}\textbf{Traffic maps.}~Fig.~\ref{fig:fig_traffic_real_corr} depicts the time series of estimated and true nominal traffic for the IPLS--CHIN OD flow (see~Fig.~\ref{fig:net_topol}(a)). For this flow, no direct NetFlow sample is collected. It is apparent that (P5) which uses the knowledge of traffic spatiotemporal correlation tracks fairly well the underlying traffic, whereas (P1) cannot even track the large-scale variations of traffic. This demonstrates the nonidentifiability of (P1) when only a small fraction $9 \%$ of OD flows are nonuniformly sampled, and notably around $10 \%$ of rows of $\bX_0$ are not directly observable. (P5) however interpolates the traffic associated with unobserved OD flows with the observed ones through the correlation matrices $\{\bR_L,\bR_Q\}$. The resulting relative estimation error for (P5) is $e_x=0.19$, which is well below $e_x=0.62$ for (P1). The correlation knowledge also helps discovering the anomalous traffic patterns as seen from Fig.~\ref{fig:fig_anomaly_map_corr}, where in particular (P5) attains $e_a=0.27$, while (P1) does $e_a=0.73$. Interestingly, the anomaly map revealed by (P1) tends to spot the anomalies intermittently since the $\ell_1$-norm regularizer weighs all flows and time-instants equally.

\begin{figure}[t]
\centering
\begin{tabular}{cc}
     \epsfig{file=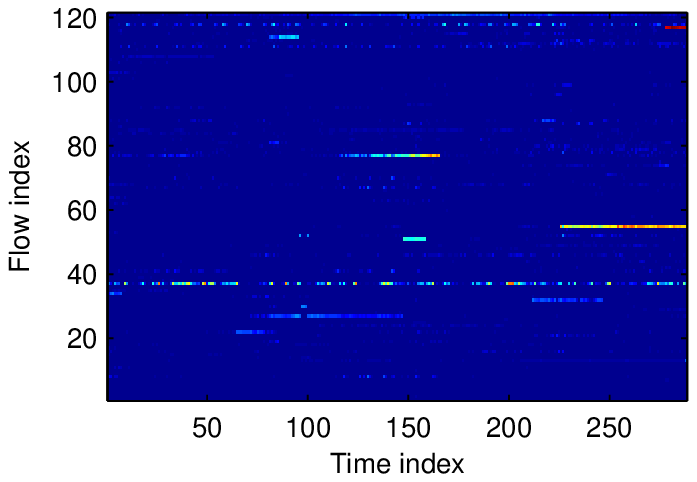,width=0.5\linewidth, height=1.25
     in } &
     \hspace{-0.75cm}\epsfig{file=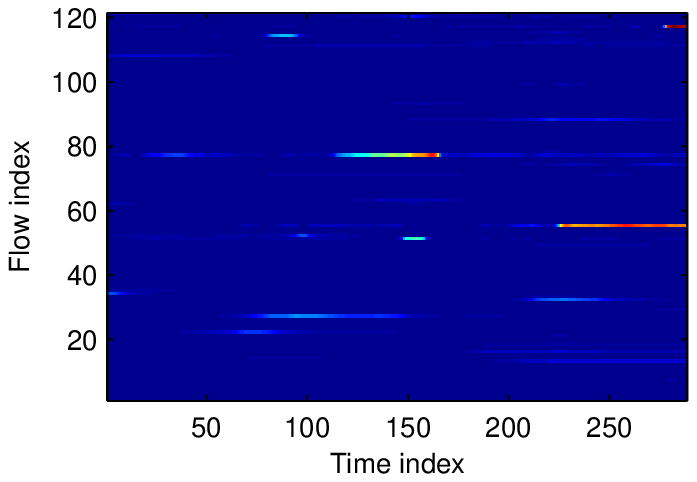,width=0.5
     \linewidth,height=1.25 in }\\
     (a) &  \hspace{-0.83cm}(b) 
\end{tabular}
\vspace{2mm}
\centering \epsfig{file=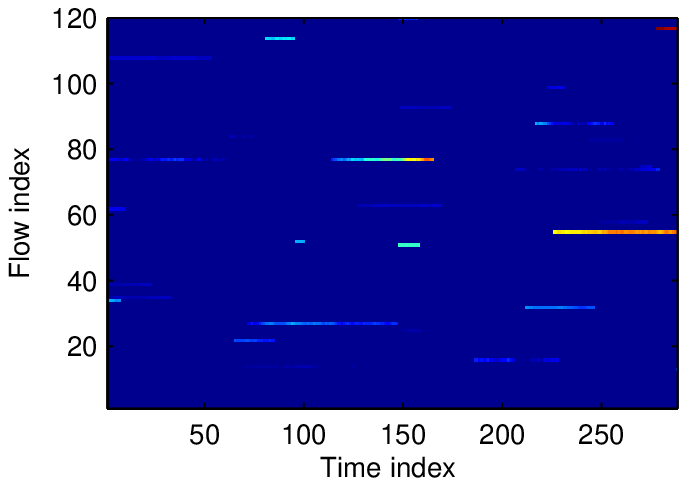,width=0.5
     \linewidth, height=1.25 in } \\
     (c)
  \caption{Estimated and ``ground truth'' (c) anomaly maps across time and flows without using correlation (a), and after using correlation information (b).}\label{fig:fig_anomaly_map_corr}
\end{figure}

\begin{figure}[t]
\centering
  \centerline{\epsfig{figure=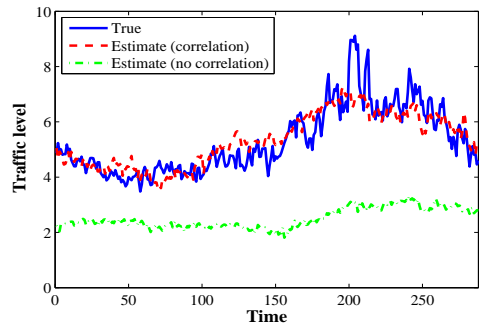,width=0.4\textwidth, height=0.25 \textwidth}}
\vspace{-2mm}\caption{True and estimated traffic of IPLS-CHIN flow.}
  \label{fig:fig_traffic_real_corr}
\end{figure}

\section{Conclusions and Future Work}
\label{sec:conclusions}
This paper taps on recent advances in low-rank and sparse recovery to create maps of nonminal and anomalous traffic as a valuable input for network management and proactive security tasks. A novel tomographic framework is put forth which subsumes critical network monitoring tasks including traffic estimation, anomaly identification, and traffic interpolation. Leveraging low intrinsic-dimensionality of nominal traffic as well as the sparsity of anomalies, a convex program is formulated with $\ell_1$- and nuclear-norm regularizers, with the link loads and a small subsets of flow counts as the available data. Under certain circumstances on the true traffic and anomalies in addition to the routing and OD-flow sampling strategies, sufficient conditions are derived, which guarantee accurate estimation of the traffic.

For practical networks where the said conditions are possibly violated, additional knowledge about inherent traffic patterns are incorporated through correlations by adopting a Bayesian approach and taking advantage of the bilinear characterization of the $\ell_1$- and nuclear-norm. A systematic approach is also devised to learn the correlations using (cyclo)stationary historical traffic data. Simulated tests with synthetic and real Internet data confirm the efficacy of the novel estimators. There are yet intriguing unanswered questions that go beyond the scope of the current paper, but worth pursuing as future research. One such question pertains to quantifying a minimal count of sampled OD flows for a realistic network scenario with a given routing matrix, which assures accurate traffic estimation. Another avenue to explore involves adoption of tensor models along the lines of~\cite{onlinetensor_2014,xiaoli_ma_anomaly,juan_tensor_tsp2013} to further exploit the network topological information toward improving the traffic estimation accuracy.

{\Large\appendix}

\section{Proof of the Main Result}
\label{sec:proof}
In what follows, conditions are first derived under which the pair $(\bX_0,\bA_0)$ is
the {\it unique} optimal solution of (P2). The sought conditions pertain to existence of certain dual certificates, which are then constructed in Section~\ref{subsec:dualcertificate}.

\subsection{Unique Optimality Conditions}
\label{subsec:uniqueoptimality}
Recall the \textit{nonsmooth} optimization problem (P2), and its Lagrangian formed as
\begin{align}
\cL(\bX,\bA;\bM_y,\bM_z)=&\|\bX\|_{\ast}+\lambda\|\bA\|_1 + \langle
\bM_y,\bY-\bR(\bX + \bA) \rangle \nonumber \\ 
& + \langle \bM_z,\bZ_{\Pi}-\cP_{\Pi}(\bX
+ \bA) \rangle \label{eq:lag_p1}
\end{align}
where $\bM_y \in \mathbb{R}^{L \times T}$ and $\bM_z \in \mathbb{R}^{F \times
T}$ are the matrices of dual variables (multipliers) associated with the link
and flow level constraints in (P2), respectively. From the characterization of
the subdifferential for the nuclear- and the $\ell_1$-norm~(see e.g.,~\cite{Boyd}),
the subdifferential of the Lagrangian at $(\bX_0,\bA_0)$ is given by (recall
that $\bX_0=\bU_0\bSigma_0\bV'_0$)
\begin{align}
&\partial_{\bX} \cL(\bX_0,\bA_0;\bM_y,\bM_z) = \left\{\bU_0 \bV_0^{'} + \bW -
\bR'\bM_y-\cP_{\Pi}(\bM_z): \right. \nonumber \\
& \hspace{2cm}  \left. \|\bW\| \leq 1,~
\cP_{\Phi_{X_0}}(\bW)=\mathbf{0}_{F \times T}   \right\} \label{eq:subdiff_x}\\
&\partial_{\bA} \cL(\bX_0,\bA_0;\bM_y,\bM_z) = \left\{\lambda {\rm sign}(\bA_0)
+ \lambda\bF - \bR' \bM_y - \cP_{\Pi}(\bM_z) : \right. \nonumber \\
& \hspace{2cm} \left. \|\bF\|_{\infty}
\leq 1, ~ \cP_{\Omega_{A_0}}(\bF)=\mathbf{0}_{F \times T}  \right\}.
\label{eq:subdiff_a}
\end{align}
The optimality conditions for (P2) assert that $(\bX_0,\bA_0)$ is an optimal
(not necessarily unique) solution if and only if
\begin{align}
 \mathbf{0}_{F \times T} \in \partial_{\bA} \cL(\bX_0,\bA_0;\bM_y,\bM_z) \nonumber \\
   \mathbf{0}_{F \times T} \in
\partial_{\bX} \cL(\bX_0,\bA_0;\bM_y,\bM_z).\nonumber
\end{align}
This is tantamount to existence of the dual variables $\{\bW,\bF,\bM_y,\bM_z\}$ satisfying: (i) $\|\bW\| \leq 1,~\cP_{\Phi_{X_0}}(\bW)=\mathbf{0}_{F \times
T}$, (ii) $\|\bF\|_{\infty} \leq 1,~\cP_{\Omega_{A_0}}(\bF)=\mathbf{0}_{F
\times T}$, and (iii) $\lambda {\rm sign}(\bA_0) + \lambda \bF = \bU \bV' +
\bW=\bR'\bM_y - \cP_{\Pi}(\bM_z)$.

In essence, to eliminate $\bM_y,\bM_z$, one can alternatively interpret iii) as finding the dual variable $\bGamma \in \cN_R^{\perp} + \cN_{\Pi}^{\perp} = (\cN_R \cap \cN_{\Pi})^{\perp}$ such that $\bGamma=\lambda {\rm sign}(\bA_0) + \lambda \bF = \bU {\bV}' + \bW$. Since $\bW=\cP_{\Phiper_{X_0}}(\bGamma)$ and $\bF=\cP_{\Omegaper_{A_0}}(\bGamma)$, conditions (i) and (ii) can also be simply recast in terms of $\bGamma$. In general, (i)--(iii) may hold for multiple solution pairs. However, the next lemma asserts that a slight tightening of the optimality conditions (i)--(iii) leads to a {\it unique} optimal solution for (P2). The proof goes along the lines of~\cite[Lemma~2]{tit_exactrecovery_2012}, and it is omitted here for conciseness.

\begin{proposition}\label{prop:prop_2}
If $(\bX_0,\bA_0)$ is locally identifiable from (c1) and (c2), and there exists a dual certificate $\bGamma \in \mathbb{R}^{F \times T}$ satisfying
\begin{align}
&\hspace{-4cm}{\rm C1)} \quad \cP_{\Phi_{X_0}}(\bGamma) = 
\bU_0\bV'_0  \nonumber\\
&\hspace{-4cm}{\rm C2)} \quad \cP_{\Omega_{A_0}}(\bGamma) = \lambda 
{\rm sgn}(\bA_0)
 \nonumber\\
&\hspace{-4cm}{\rm C3)} \quad \cP_{\cN_R \cap \cN_{\Pi}} 
(\bGamma) = \mathbf{0} \nonumber\\
&\hspace{-4cm}{\rm C4)} \quad \|\cP_{\Phiperx}(\bGamma)\| < 1
\nonumber\\
&\hspace{-4cm}{\rm C5)} \quad \|\cP_{\Omegapera}(\bGamma)\|_{\infty}
< \lambda \nonumber
\end{align}
then $(\bX_0,\bA_0)$ is the unique optimal solution to (P2).
\end{proposition}

The rest of the proof deals with construction of a valid dual certificate $\bGamma$ that simultaneously meets C1--C5.

One should note that condition (iii) is a distinct feature of the recovery task pursued in this paper. In a similar context, in the robust PCP problem studied in~\cite{CSPW11}, $\cN_R=\emptyset,\cN_{\Pi}=\emptyset$, and thus C3 does not appear anymore. In addition, the low-rank plus compressed sparse recovery task studied in~\cite{tit_exactrecovery_2012} does not involve the intersection of subspaces as appearing in C3.

\subsection{Dual Certificate Construction}
\label{subsec:dualcertificate}
The main steps of the construction are inspired by~\cite{CSPW11} which studies decomposition of low-rank plus sparse matrices, that is, $\Pi=\emptyset$ and $\bR=\bI_F$. However, relative to~\cite{CSPW11} the 
problem here brings up several new distinct elements including the null space 
of compression and sampling operators in C3, which further challenge construction of 
dual certificates, and demands, in part, a new treatment. In addition, different incoherence measures are introduced here which facilitate satisfiability for random ensembles. The construction 
involves two steps. In the first step, a candidate dual
certificate is selected to fulfil C1--C3, whereas the second step assures
the candidate dual certificate satisfies C4--C5 as well under certain
technical conditions in terms of the incoherence parameters in Section~\ref{subsec:incoherence_measures}.

Toward the first step, condition (II) in Theorem~\ref{th:theorem_1} implies local identifiability of the observation model, namely $\Omega_{A_0} \cap \Phi_{X_0}=\{\mathbf{0}\}$ and $(\Omega_{A_0} \oplus \Phi_{X_0}) \cap (\cN_R \cap \cN_{\Pi}) = \{\mathbf{0}\}$, and thus based on a property 
of direct-sum~\cite{Horn} there {\it exists} a {\it unique} certificate 
$\bGamma \in \Omega_{A_0} \oplus \Phi_{X_0} \oplus (\cN_R \cap 
\cN_{\Omega})$ with projections $\cP_{\Omega_{A_0}}(\bGamma)=\lambda {\rm 
sign}(\bA_0)$, $\cP_{\Phi_{X_0}}(\bGamma)= \bU_0\bV'_0$, and $\cP_{\cN_R \cap 
\cN_{\Pi}}(\bGamma)= \mathbf{0}$. This dual certificate can be
expressed as $\bGamma=\bGamma_{\Omega_{A_0}}+\bGamma_{\Phi_{X_0}} + 
\bGamma_{\cN_R \cap \cN_{\Pi}}$ with the components $\bGamma_{\Omega_{A_0}} 
\in \Omega_{A_0}$, $\bGamma_{\Phi_{X_0}}
\in \Phi_{X_0}$, and $\bGamma_{\cN_R \cap \cN_{\Pi}}
\in \cN_R \cap \cN_{\Pi}$. As will be seen later, it is more convenient to 
represent
$\bGamma_{\Omega_{A_0}}=\epsilon_{\Omega_{A_0}} + \lambda {\rm sign}(\bA_0) $
and
$\bGamma_{\Phi_{X_0}}=\epsilon_{\Phi_{X_0}}+\bU_0\bV'_0$. From
C1--C3, for the projection components
$\{\epsilon_{\Omega_{A_0}},\epsilon_{\Phi_{X_0}},\bGamma_{\cN_R \cap 
\cN_{\Pi}}\}$ it then holds that
\begin{align}
&\epsilon_{\Phi_{X_0}} = - \cP_{\Phi_{X_0}}( \epsilon_{\Omega_{A_0}} +  \lambda
{\rm sign}(\bA_0) + \bGamma_{\cN_R \cap \cN_{\Pi}})  \label{eq:eps_phi} \\
&\epsilon_{\Omega_{A_0}} = - \cP_{\Omega_{A_0}}( \epsilon_{\Phi_{X_0}} +
\bU_0\bV'_0 + \bGamma_{\cN_R \cap \cN_{\Pi}})  \label{eq:eps_omega}\\
&\bGamma_{\cN_R \cap \cN_{\Pi}} \nonumber \\ 
&= -\cP_{\cN_R \cap 
\cN_{\Pi}}(\epsilon_{\Phi_{X_0}} +
\bU_0\bV'_0 + \epsilon_{\Omega_{A_0}} +  \lambda
{\rm sign}(\bA_0) ).  \label{eq:gamma_null}
\end{align}
%

The second step of the proof manages the candidate dual certificate $\bGamma$ to satisfy C4 and C5 as well. The main idea is to tighten the conditions for 
local identifiability, and impose additional conditions on the incoherence 
measures (c.f. Section~\ref{subsec:incoherence_measures}) to ensure that C4 and C5 hold true. In this direction, one can begin by bounding
\begin{align}
\|\cP_{\Phiperx}(\bGamma)\| &\leq \|\bGamma_{\Omega_{A_0}} + \bGamma_{\cN_R 
\cap \cN_{\Omega}}\|  \nonumber \\ 
&=
\|\epsilon_{\Omega_{A_0}}+\lambda {\rm sgn}(\bA_0) + \bGamma_{\cN_R \cap 
\cN_{\Omega}}\|  \nonumber\\
&\stackrel{(a)}{\leq} \|\cP_{\Omega_{A_0}}(\epsilon_{\Phi_{X_0}} +
\bU_0\bV'_0 + \bGamma_{\cN_R \cap \cN_{\Omega}} )\| \nonumber \\ 
&\hspace{0.25cm}+ \lambda
\|{\rm sgn}(\bA_0)\| + \|\bGamma_{\cN_R \cap \cN_{\Omega}}\| 
\label{eq:bnd_projphipergamma}
\end{align}
and
\begin{align}
\|\cP_{\Omegapera}(\bGamma)\|_{\infty} &\leq \|\bGamma_{\Phi_{X_0}} + 
\bGamma_{\cN_R \cap 
\cN_{\Pi}} \|_{\infty} \nonumber \\ 
&=
\|\epsilon_{\Phi_{X_0}}+ \bU_0\bV_0^{'} + \bGamma_{\cN_R \cap 
\cN_{\Pi}}\|_{\infty}  \nonumber\\
&\stackrel{(b)}{\leq} \|\cP_{\Phi_{X_0}}(\epsilon_{\Omega_{A_0}} + \lambda
{\rm sgn}(\bA_0) + \bGamma_{\cN_R \cap \cN_{\Pi}})\|_{\infty}  \nonumber \\ 
& \hspace{0.25cm} +
\|\bU_0\bV_0^{'}\|_{\infty} + \|\bGamma_{\cN_R \cap \cN_{\Pi}}\|_{\infty} 
\label{eq:bnd_projomegapergamma}
\end{align}
where (a) and (b) come from~\eqref{eq:eps_phi} and~\eqref{eq:eps_omega} after applying the triangle inequality. In order to bound the r.h.s. 
of~\eqref{eq:bnd_projphipergamma} and~\eqref{eq:bnd_projomegapergamma}, it is
instructive first to recognize that $\|\bU_0\bV_0^{'}\|_{\infty} \leq \gamma(\bU_0,\bV_0)$, and
\begin{align}
\|{\rm sgn}(\bA_0)\| \leq \left( \|{\rm sgn}(\bA_0)\|_{\infty,\infty} \|{\rm sgn}(\bA_0)\|_{1,1}  \right)^{1/2} = k
\end{align}
see e.g.,~\cite{Horn}. In addition, building on~\eqref{eq:incoherence_add} and~\eqref{eq:tau_infty}, the first term in the r.h.s. of~\eqref{eq:bnd_projphipergamma} is bounded as
\begin{align}
&\|\cP_{\Omega_{A_0}}(\epsilon_{\Phi_{X_0}} +
\bU_0\bV_0^{'} + \bGamma_{\cN_R \cap 
\cN_{\Pi}})\| \nonumber \\
& \leq \|\cP_{\Omega_{A_0}} \cP_{\Phi_{X_0}}(\epsilon_{\Phi_{X_0}} + \bU_0\bV_0^{'})\| \nonumber \\ 
&\hspace{2cm} + \|\cP_{\Omega_{A_0}}\cP_{\cN_{\Pi}} \cP_{\cN_R}(\bGamma_{\cN_R \cap 
\cN_{\Pi}})\| \nonumber\\
& \stackrel{(a)}{\leq}  \mu(\Phi_{X_0},\Omega_{A_0}) \left( \|\epsilon_{\Phi_{X_0}}\| + 
1 \right) \nonumber \\ 
&\hspace{2cm} + \mu(\cN_R,\Omega_{A_0} \cap \cN_{\Pi}) \|\bGamma_{\cN_R \cap 
\cN_{\Omega}}\| \label{eq:bnd_term1_rhs}
\end{align}
where (a) is due to the fact that $\cP_{\Omega_{A_0} \cap \cN_{\Pi} } = \cP_{\Omega_{A_0}} \cP_{\cN_{\Pi}}$.



Proceeding in a similar manner as for~\eqref{eq:bnd_term1_rhs}, upon using~\eqref{eq:gamma_uv} it follows that
\begin{align}
&\|\cP_{\Phi_{X_0}}(\epsilon_{\Omega_{A_0}} + \lambda
{\rm sgn}(\bA_0) + \bGamma_{\cN_R \cap 
\cN_{\Pi}} )\|_{\infty} \nonumber \\ 
&\leq  \gamma(\bU_0,\bV_0) 
~\|\cP_{\Phi_{X_0}}(\epsilon_{\Omega_{A_0}} + \lambda {\rm sgn}(\bA_0) + 
\bGamma_{\cN_R \cap \cN_{\Pi}})\|  \nonumber  \\
& \leq \gamma(\bU_0,\bV_0) 
~\big[ \|\cP_{\Phi_{X_0}}(\epsilon_{\Omega_{A_0}} + \lambda {\rm sgn}(\bA_0) \| \nonumber \\ 
& \hspace{3cm} + \|\cP_{\Phi_{X_0}}\cP_{\cN_{\Pi}}(\bGamma_{\cN_R \cap \cN_{\Pi}})\|\big]  \nonumber  \\
& \leq \gamma(\bU_0,\bV_0) \big[ \mu(\Omega_{A_0},\Phi_{X_0}) \left( \|\epsilon_{\Omega_{A_0}}\| + \lambda k \right) \nonumber \\ 
& \hspace{3cm}+  \mu(\Phi_{X_0},\cN_{\Pi}) \|\bGamma_{\cN_R \cap \cN_{\Pi}}\| \big].  \label{eq:bnd_term2_rhs}
\end{align}

Focusing on~\eqref{eq:bnd_term2_rhs} and~\eqref{eq:bnd_term1_rhs}, it is only 
left to bound $\|\epsilon_{\Omega_{A_0}}\|$, $\|\epsilon_{\Phi_{X_0}}\|$, 
and $\|\bGamma_{\cN_R \cap 
\cN_{\Omega}}\|$. To this end, ~\eqref{eq:eps_omega}-\eqref{eq:gamma_null} 
are utilized to arrive at
\begin{align}
&\|\epsilon_{\Phi_{X_0}}\|  = \| \cP_{\Phi_{X_0}}( \epsilon_{\Omega_{A_0}} 
+  \lambda
{\rm sign}(\bA_0) + \bGamma_{\cN_R \cap 
\cN_{\Pi}}  ) \|  \nonumber\\
&\leq \mu(\Phi_{X_0},\Omega_{A_0}) \left( 
\|\epsilon_{\Omega_{A_0}}\| + \lambda k  \right) + \mu(\Phi_{X_0},\cN_{\Pi})
\|\bGamma_{\cN_R \cap \cN_{\Pi}}\|
\label{eq:bnd_eps_phi}
\end{align}
\begin{align}
&\| \epsilon_{\Omega_{A_0}} \| = \| \cP_{\Omega_{A_0}}( 
\epsilon_{\Phi_{X_0}} + \bU_0\bV'_0 + \bGamma_{\cN_R \cap 
\cN_{\Pi}} ) \| \nonumber\\
&\leq \mu(\Phi_{X_0},\Omega_{A_0}) \left( 
\|\epsilon_{\Phi_{X_0}}\| + 1 \right) + \mu(\cN_R,\Omega_{A_0}) \|\bGamma_{\cN_R \cap 
\cN_{\Omega}} \| \label{eq:bnd_eps_omega}
\end{align}
and
\begin{align}
&\|\bGamma_{\cN_R \cap \cN_{\Omega}}\| \nonumber \\
& = \| \cP_{\cN_R \cap 
\cN_{\Omega}}(\epsilon_{\Phi_{X_0}} +
\bU_0\bV'_0 + \epsilon_{\Omega_{A_0}} +  \lambda
{\rm sign}(\bA_0) ) \| \nonumber\\
&\leq \mu(\Phi_{X_0},\cN_{\Omega}) \big(\|\epsilon_{\Phi_{X_0}}\|  + 1 \big) \nonumber \\ 
& \hspace{2cm}+ \mu(\cN_R,\Omega_{A_0} \cap \cN_{\Omega}) \big (\|\epsilon_{\Omega_{A_0}} \|
+ \lambda k ).\label{eq:bnd_gamma_null}
\end{align}
For convenience introduce the notations $\alpha:= \mu(\Phi_{X_0},\Omega_{A_0})$, $\beta:=\mu(\cN_R,\Omega_{A_0})$, $\xi:=\mu(\Phi_{X_0},\cN_{\Pi})$, and $\nu:=\mu(\cN_R,\Omega_{A_0} \cap \cN_{\Pi})$. Then, after mixing~\eqref{eq:bnd_eps_phi}--\eqref{eq:bnd_gamma_null} and doing some algebra it follows that
\begin{align}
\| \Gamma_{\cN_R \cap \cN_{\Omega}} \| \leq \theta:= \frac{\xi +\lambda k \nu + \alpha (\xi+\alpha \nu)(1-\alpha^2)(\alpha + \lambda k) }{1-\nu \beta - (\xi+\alpha \nu)(1-\alpha^2)(\xi+\alpha \beta)} \label{eq:bnd_norm_gamma}
\end{align}
and
\begin{align}
\| \epsilon_{\Omega_{A_0}} \| \leq & \alpha+(1-\alpha^2) \alpha^2 (\alpha+\lambda k) \nonumber \\ 
& + \big[\beta + \alpha^2 (1-\alpha^2)\beta + \alpha \xi (1-\alpha^2)^{-1} \big] \theta \label{eq:bnd_norm_eps_omega}
\end{align}
\begin{align}
\| \epsilon_{\Phi_{X_0}} \| \leq (1-\alpha^2) \big[ \alpha (\alpha + \lambda k) + (\alpha \beta + \xi) \theta \big].  \label{eq:bnd_norm_eps_phi}
\end{align}
At this point, it is important to recognize from~\eqref{eq:tau_infty} that $\|\bGamma_{\cN_R \cap 
\cN_{\Pi}} \|_{\infty} \leq \tau \|\bGamma_{\cN_R \cap 
\cN_{\Pi}} \|$.

Now building on \eqref{eq:bnd_norm_gamma}-\eqref{eq:bnd_norm_eps_phi}, one can bound the terms in the r.h.s. of~\eqref{eq:bnd_projphipergamma}~and~\eqref{eq:bnd_projomegapergamma} from above in terms of $\{\alpha,\beta,\xi,\nu,k\}$. Finally, to fulfill C4 and C5, it suffices to confine their corresponding upper bounds to the values $1$ and $\lambda$, respectively. This imposes the conditions

\vspace{0.2cm}
\hspace{-0.25cm}(a) $\lambda k + \alpha + \alpha (1-\alpha^2) \big[ \alpha (\alpha + \lambda k) + (\alpha \beta + \xi) \theta \big] + (1+\nu) \theta < 1$

\vspace{0.05cm}
\hspace{-0.25cm}(b) $\gamma + \eta \alpha \lambda k + (\tau + \eta \alpha + \eta \xi) \theta < \lambda $.
\vspace{0.2cm}

The conditions (a) and (b) imply that C1--C5 hold for the dual certificate $\bGamma$ if there exists a valid $\lambda \in [\lambda_{\min},\lambda_{\max}]$, with $\lambda_{\max} \geq \lambda_{\min} \geq 0$. The resulting condition is then summarized in the assumptions (I) and (II) of Theorem~\ref{th:theorem_1}, and the proof is now complete.






\end{document}